\documentclass[11pt, final]{article}
\usepackage{amsthm,amsmath,amssymb,booktabs,xcolor,graphicx}
\usepackage{bbm}
\usepackage{listings}
\usepackage{xcolor}
\usepackage{appendix}
\usepackage{enumerate}
\usepackage[shortlabels]{enumitem}
\usepackage{verbatim}
\usepackage[ruled,vlined]{algorithm2e}
\usepackage{algorithmic}

\usepackage[margin=1in]{geometry}
\usepackage[utf8]{inputenc}
\usepackage[T1]{fontenc}

\usepackage[textsize=small,backgroundcolor=orange!20]{todonotes}
\usepackage[colorlinks=true, allcolors=blue]{hyperref}
\usepackage{url}
\usepackage{etoolbox}
\usepackage{appendix}
\usepackage[nameinlink, noabbrev,capitalize]{cleveref}
\crefname{equation}{}{} 
\AtBeginEnvironment{appendices}{\crefalias{section}{appendix}} 

\usepackage[color]{showkeys} 

\colorlet{refkey}{orange!20}
\colorlet{labelkey}{blue!30}

\numberwithin{equation}{section}
\newtheorem{theorem}{Theorem}[section]
\newtheorem{proposition}[theorem]{Proposition}
\newtheorem{lemma}[theorem]{Lemma}

\crefname{claim}{Claim}{Claims}

\newtheorem{corollary}[theorem]{Corollary}

\newtheorem*{question*}{Question}

\theoremstyle{definition}
\newtheorem{definition}[theorem]{Definition}

\newtheorem*{definition*}{Definition}

\theoremstyle{remark}
\newtheorem*{remark}{Remark}

\newcommand{\abs}[1]{\left\lvert#1\right\rvert}

\newcommand{\snorm}[1]{\lVert#1\rVert}

\newcommand{\mb}{\mathbb}
\newcommand{\mbf}{\mathbf}

\newcommand{\mc}{\mathcal}
\newcommand{\mf}{\mathfrak}

\newcommand{\on}{\operatorname}
\newcommand{\wh}{\widehat}
\newcommand{\wt}{\widetilde}

\allowdisplaybreaks

\allowdisplaybreaks

\author{
Vishesh Jain \\
Stanford University \\
\texttt{vishesh.vj@gmail.com}\\
\and
Natesh S. Pillai\\
Harvard University\\
\texttt{pillai@fas.harvard.edu}
\and
Ashwin Sah \\
Massachusetts Institute of Technology \\
\texttt{asah@mit.edu}
\and
Mehtaab Sawhney \\
Massachusetts Institute of Technology \\
\texttt{msawhney@mit.edu}
\and
Aaron Smith\\
University of Ottawa \\
\texttt{asmi28@uottawa.ca}
}

\date{}
\begin{document}
\title{Fast and memory-optimal dimension reduction using Kac's walk}
\begin{titlepage}
\clearpage\maketitle
\thispagestyle{empty}
\begin{abstract}
In this work, we analyze dimension reduction algorithms based on the Kac walk and discrete variants.
\begin{itemize}
    \item For $n$ points in $\mb{R}^{d}$, we design an optimal Johnson-Lindenstrauss (JL) transform based on the Kac walk which can be applied to any vector in time $O(d\log{d})$ for essentially the same restriction on $n$ as in the best-known transforms due to Ailon and Liberty [SODA, 2008], and Bamberger and Krahmer [arXiv, 2017]. 
    Our algorithm is memory-optimal, and outperforms existing algorithms in regimes when $n$ is sufficiently large and the distortion parameter is sufficiently small. In particular, this confirms a conjecture of Ailon and Chazelle [STOC, 2006] in a stronger form.
    \item The same construction gives a simple transform with optimal Restricted Isometry Property (RIP) which can be applied in time $O(d\log{d})$ for essentially the same range of sparsity as in the best-known such transform due to Ailon and Rauhut [Discrete Comput.\ Geom., 2014].
    \item We show that by fixing the angle in the Kac walk to be $\pi/4$ throughout, one obtains optimal JL and RIP transforms with almost the same running time, thereby confirming -- up to a $\log\log{d}$ factor -- a conjecture of Avron, Maymounkov, and Toledo [SIAM J.\ Sci.\ Comput., 2010]. Our moment-based analysis of this modification of the Kac walk may also be of independent interest.
\end{itemize}

\end{abstract}
\end{titlepage}
\newpage

\section{Introduction}\label{sec:introduction}
The aim of this paper is to design fast and simple dimensionality reduction algorithms with optimal embedding dimension -- specifically, fast Johnson-Lindenstrauss (JL) transforms and fast Restricted Isometry Property (RIP) transforms -- using the Kac walk and some of its discrete variants.

\subsection{Fast Johnson-Lindenstrauss Transforms (FJLTs)}
The classical lemma of Johnson and Lindenstrauss \cite{JL84} asserts that for any collection of $n$ points $x_1,\dots,x_{n}$ in Euclidean space $\mb{R}^{d}$, and for any error parameter $\epsilon \in (0,1)$, there exists a linear transformation $\Phi : \mb{R}^{d}\to \mb{R}^{k}$, with $k = O(\epsilon^{-2}\log{n})$, such that for all $i\in [n]$, $\|\Phi x_{i}\|_{2} = (1\pm \epsilon)\cdot \|x_{i}\|_{2}$. At least for $ \epsilon > d^{-0.49}$, the bound on $k$ is known to be optimal up to constants (\cite{LN17, AK17}).

The early examples of optimal JL embeddings (i.e. JL embeddings with asymptotically optimal embedding dimension) are (suitably rescaled) random Gaussian \cite{dasgupta2003elementary} and random Rademacher matrices \cite{Ach03}. While achieving the optimal embedding dimension for essentially all settings of the parameters $d,n,\epsilon$, such embeddings are unfortunately too slow for many applications, since the time to compute the image $\Phi v$ of a fixed vector $v \in \mb{R}^{d}$ is in general $O(dk) = O(\epsilon^{-2}d\log{n})$. 

To address this issue, optimal JL embeddings for which the image $\Phi v$ of a fixed vector $v\in \mb{R}^{d}$ can be computed in time $O(d\log{d})$ (we will often refer to this as the \emph{running time}), under some restrictions on $n$ and $\epsilon$, have been proposed, starting with the seminal work of Ailon and Chazelle \cite{AC09}, who constructed a family of optimal JL embeddings with running time
\begin{align}
\label{eq:ac-time}
    O\left(d\log{d} + \min\{\epsilon^{-2}d\log{n}, \epsilon^{-2}\log^{3}{n}\}\right).
\end{align}
In particular, for the $d\log{d}$ term to dominate the second term, we must have 
\begin{align}
\label{eqn:ac}
n \le \exp(\tilde{O}(\epsilon^{2/3}d^{1/3})),
\end{align}
where $\tilde{O}$ hides possible logarithmic factors in $\epsilon$ and $d$. 
At least for fixed $\epsilon$, this restriction was significantly relaxed by Ailon and Liberty \cite{ailon2009fast}, who provided a different family of optimal JL embeddings (for any $\gamma > 0$) with runtime
\begin{align}
\label{eqn:al}
O(d\log({\epsilon^{-2}\log{n}})) \quad \text{for all } n \le \exp(O_{\gamma, \epsilon}(d^{1/2 - \gamma})).
\end{align}
In a recent work of Bamberger and Krahmer \cite{bamberger2017optimal}, an optimal JL embedding, which is simpler than the construction in \cite{ailon2009fast}, is provided with runtime
\begin{align}
\label{eqn:bk}
O(d\log(\epsilon^{-2}\log{n})) \quad \text{for all }  n \le \exp(\tilde{O}(\epsilon^{2}d^{1/2})).
\end{align}
Note that for the regime not covered by \cref{eqn:ac}, the running time of $O(d\log(\epsilon^{-2}\log{n}))$ in \cref{eqn:al,eqn:bk} simplifies to $O(d\log{d})$ as well.  

Finally, we note that there is a separate line of work focused on designing optimal JL embeddings with even faster running times on sparse vectors; since this is not the focus of the present work, we omit further discussion, and refer the reader to \cite{KN14}, noting only that these sparse JL transforms may be used to improve the first term inside the $\min$ in \cref{eq:ac-time} to $\epsilon^{-1}d\log{n}$.

\subsection{The Restricted Isometry Property (RIP) and fast RIP transforms.}The design of JL embeddings with running time $O(d\log{d})$ (albeit with suboptimal embedding dimension) for $n = \exp(\omega_{\epsilon}(\sqrt{d}))$ is based on the connection between JL transforms and transforms satisfying the Restricted Isometry Property (RIP). We recall this important notion, which was first isolated in the compressed sensing literature \cite{candes2006robust, donoho2006compressed}.

\begin{definition}
For a matrix $A$, define
\[\delta_s(A) = \sup_{\substack{\snorm{x}_{2}=1\\x\text{ is }s\text{-sparse}}}|\snorm{Ax}_2^2-1|.\]
We say that $A$ has the \emph{Restricted Isometry Property (RIP) of order $s$ and level $\delta$} if $\delta_s(A)\le\delta$.
\end{definition}

\begin{remark}
It is known (see, e.g., \cite{baraniuk2008simple}) that for any $k\times d$ matrix $A$, $\delta_s(A)\gtrsim\sqrt{(s\log (d/s))/k}$. Hence, we will informally say that a $k\times d$ matrix $A$ is \emph{RIP-optimal at $s$} if
\begin{equation}
\label{eqn:optimal-RIP}
\delta_s(A)\lesssim\sqrt{\frac{s \log (d/s)}{k}}.
\end{equation}
\end{remark}

As in the case of optimal JL transforms, the early constructions of optimal RIP transforms are based on random subgaussian matrices (see, e.g., \cite{baraniuk2008simple}). Once again, these transforms have the drawback of not supporting fast matrix-vector multiplication, leading to the study of fast (nearly) optimal RIP transforms i.e. $k\times d$ matrices supporting matrix-vector multiplication in time $O(d\log{d})$, and which satisfy the RIP property of order $s$ and level $\delta$ with $s,\delta,k$ related (possibly up to polylogarithmic factors in $d$) as in \cref{eqn:optimal-RIP}. 

Notably, improving on previous work of Candes and Tao \cite{CT06}, Rudelson and Vershynin \cite{RM08} showed that a (suitably rescaled) random sample of $k = \Omega(\delta^{-2}s\log^{4}d)$ rows of the Walsh-Hadamard matrix satisfies, with high probability, the RIP of order $s$ and level $\delta$. Since the Walsh-Hadamard matrix supports time $O(d\log{d})$ matrix-vector multiplication via the Fast Walsh-Hadamard Transform, this gives a fast, nearly-optimal (in terms of embedding dimension) RIP transform. The result of Rudelson and Vershynin is optimal up to a factor of $\log^{3}d$, which has since been improved (at least if one is willing to allow a slightly worse dependence on $\delta$) -- see \cite{haviv2017restricted} for an account of these developments.\\

There is a certain sense in which optimal RIP and optimal JL transforms are nearly equivalent. Indeed, an $\epsilon$-net argument shows (see \cite{baraniuk2008simple} for details) that optimal JL embeddings are also optimal RIP embeddings. In particular, this shows that the fast JL embedding in \cref{eqn:al} gives a fast optimal RIP transform for
\begin{align*}
    s \le O_{\gamma,\delta}(d^{1/2-\gamma}).
\end{align*}
In later work of Ailon and Rauhut \cite{AR14}, a simpler fast optimal RIP transform was obtained for
\begin{align}
\label{eqn:ah}
s \le \tilde{O}(\delta d^{1/2}).     
\end{align}
We also note that the optimal JL transform in \cref{eqn:bk} can be used to obtain an even simpler fast optimal RIP transform up to 
\begin{align}
\label{eqn:bk-rip}
    s \le \tilde{O}(\delta^{2}d^{1/2}),
\end{align}
although this connection does not seem to have been observed in \cite{bamberger2017optimal}.\\

In the other direction, a remarkable result of Krahmer and Ward \cite{KW11} (see also \cref{thm:krahmer-ward} below) shows that for any $k\times d$ matrix $A$ with RIP of order $s$ and level $\delta/4$, the random matrix $AD$, where $D$ is a random diagonal Rademacher matrix, satisfies (with high probability) the JL property for a given collection of $n$ points with error $\delta$, provided that $n \le 2^{s}$. This result will prove to be crucial for us.   

\subsection{The Kac walk and Orthogonal Repeated Averaging (ORA)}
Introduced by Mark Kac \cite{kac1956foundations} in 1956 as a toy model for a one-dimensional Boltzmann gas, the Kac walk is the following discrete time Markov chain $\{Q_t\}_{t\ge 0}$ on the special orthogonal group $\on{SO}(d)$.   

\begin{definition}
Let $Q_0 = I_d\in\mb{R}^{d\times d}$. For all integers $t\ge 1$, sample two distinct uniform random coordinates $i_t,j_t\in[d]$ and a uniform random angle $\theta_t$ from $[0,2\pi)$. Then, let $Q_t = R_{i_t,j_t,\theta_t}Q_{t-1}$, where $R_{i,j,\theta}\in\mb{R}^{d\times d}$ is the rotation in the $(i,j)$ plane given by:
\begin{align*}
R_{i,j,\theta}(e_k) &= e_k \quad \text{for all $k \notin \{i,j\}$};\\
R_{i,j,\theta}(x_ie_i+x_je_j) &= (x_i\cos\theta - x_j\sin\theta)e_i + (x_i\sin\theta + x_j\cos\theta)e_j.
\end{align*}
By the \emph{Kac walk of length} $T$ we mean the random variable $Q_T$.
\end{definition}

The Kac walk has a rich history in probability and mathematical physics (see, e.g., the references in \cite{PS17, PS18}). Its utility for dimensionality reduction was first suggested by Ailon and Chazelle \cite{AC09}, who also noted that the Kac walk has the attractive property that given the update sequence $\{R_{i_t, j_t, \theta_t}\}_{t\in [T]}$, the image $Q_{T}v$ of any vector $v \in \mb{R}^{d}$ can be computed with only a \emph{constant} amount of memory overhead. Ailon and Chazelle conjectured that the Kac walk performs at least as well as fast JL transforms based on the fast Walsh-Hadamard transform. Specifically, they conjectured that for a given set of $n$ points in $\mb{R}^{d}$ and error parameter $\epsilon$, projecting $Q_T$ onto the first $O(\epsilon^{-2}\log{n})$ coordinates gives a JL embedding of the point set with relative error $\epsilon$, provided that
$$T = O(d\log{d} + \text{poly}(\log{n},\epsilon^{-1})).$$
Recently, Choromanski, Rowland, Chen, and Weller \cite{choromanski2019unifying} provided numerical support for this conjecture. 

Despite this numerical evidence, the conjecture of Ailon and Chazelle  may perhaps seem quite surprising from the point of view of mixing times of Markov chains. The initial proof of the JL lemma due to Johnson and Lindenstrauss \cite{JL84} is based on taking the embedding matrix to be a uniformly random sample from the $O(\epsilon^{-2}\log{n})$-Stiefel manifold in $d$-dimensions (i.e.\ the uniform distribution over the set of $O(\epsilon^{-2}\log{n})$-orthonormal frames in $\mb{R}^{d}$). On the other hand, dimensional considerations show that the Kac walk \emph{does not} mix on the $k$-Stiefel manifold in $d$-dimensions in $\Omega(kd)$ steps (see \cite[Theorem~6]{Oli09} for a formal proof); in our setting, this would give a lower bound of $\Omega(\epsilon^{-2}d\log{n})$, which asymptotically matches multiplication by a $O(\epsilon^{-2}\log{n})\times d$ random Gaussian matrix. Indeed, Oliveira \cite{Oli09} conjectured that the Kac walk may be used to design JL transforms running in time $O(\epsilon^{-2}(\log{n})d\log{d})$ (this is slower than multiplication by a Gaussian matrix, but only requires a constant memory overhead) based on this connection with mixing on the $O(\epsilon^{-2}\log{n})$-Stiefel manifold in $d$-dimensions.\\ 

Nevertheless, as our first main result (\cref{thm:uniform-main}), we confirm the conjecture of Ailon and Chazelle in a much stronger form by showing that the mixing of the Kac walk on the $1$-Stiefel manifold in $d$-dimensions (i.e.\ the unit sphere) is already enough for the purpose of dimensionality reduction -- while simple in hindsight, we believe that this provides a more intuitive and principled explanation for the existence of JL transforms running in time $O(d\log{d})$ than is obtained from the analysis transforms based on Hadamard matrices. Specifically, we provide a fast and memory-optimal JL transform based on the Kac walk running in time
\begin{align*}
    T = O\left(d\log{d} + \min\{d\log{n}, \epsilon^{-2}\log^{2}{n}\log^{2}(\log{n})\log^{3}{d}\}\right).
\end{align*}
In particular, the first term dominates the running time provided that
\begin{align*}
    n \le \exp(\tilde{O}(\epsilon d^{1/2})),
\end{align*}
which matches the regime covered by \cref{eqn:al,eqn:bk} in terms of $d$ (up to polylogarithmic factors), and has improved dependence on the error parameter $\epsilon$ (in particular, the first term inside the $\min$ is $d\log{n}$, which improves on the previously best known rate of $\epsilon^{-1}d\log{n}$ obtained by sparse JL transforms). As a corollary (\cref{cor:RIP-optimal}), we also obtain a fast RIP transform, much simpler than in \cite{AR14}, up to 
\begin{align*}
    s \le \tilde{O}(\delta d^{1/2}),
\end{align*}
which matches the restriction in \cref{eqn:ah} up to polylogarithmic factors in $d$ and is better in terms of $\delta$ dependence than \cref{eqn:bk-rip}.\\ 

For the purpose of computing matrix-vector products, an even faster and more elegant approach is to fix the angle in the Kac walk for the entire process to be $\theta = \pi/4$. This leads to the following discrete time Markov chain $\{Q_{t}\}_{t\ge 0}$ on $\on{SO}(d)$, which we call orthogonal repeated averaging (ORA) due to its apparent similarity to various iterated averaging processes in the probability literature (see, e.g., \cite{chatterjee2019phase} and the references therein).   

\begin{definition}
Let $Q_0 = I_d\in\mb{R}^{d\times d}$. For all integers $t\ge 1$, sample two distinct uniform random coordinates $i_t,j_t\in[d]$, and let $Q_t = R_{i_t,j_t}Q_{t-1}$, where $R_{i,j}\in\mb{R}^{d\times d}$ is the rotation in the $(i,j)$ plane given by:
\begin{align*}
R_{i,j,\theta}(e_k) &= e_k \quad \text{for all $k \notin \{i,j\}$};\\
R_{i,j,\theta}(x_ie_i+x_je_j) &= \left(\frac{x_i + x_j}{\sqrt{2}}\right)e_i + \left(\frac{x_i - x_j}{\sqrt{2}}\right)e_j.
\end{align*}
By the \emph{ORA of length} $T$ we mean the random variable $Q_T$.
\end{definition}

The application of ORA to dimensionality reduction was suggested by Avron, Maymounkov, and Toledo \cite{AMT10}, who conjectured (based on experimental evidence) that the ORA performs as well as the Kac walk for dimensionality reduction. From the point of view of mixing times, this is even more delicate since, for instance, the total variation distance between the uniform distribution on the sphere and the ORA distribution for \emph{any} finite number of steps is always $1$. In our final main result (\cref{thm:discrete-main}), we almost confirm this conjecture by designing an optimal JL transform based on ORA running in time
\begin{align*}
    T = O\left(d\log{d}\log\log{d} + \log{d}\min\{d\log{n}, \epsilon^{-2}\log^{2}n \log^{2}(\log{n})\log^{3}{d} \}\right).
\end{align*}
As in \cref{cor:RIP-optimal}, this also gives an RIP-optimal transform. 
\begin{remark}
We conjecture that the additional $\log\log{d}$ factor (which is anyway essentially constant for practical purposes) can be removed, and expect the ORA based transform to be more efficient than the Kac walk based transformed in practice. 
\end{remark}

We now proceed to a formal statement of our main results.

\subsection{Main Results}
As mentioned above, our first result is a fast JL transform based on the Kac walk
which essentially matches the fastest known JL transforms based on subsampled Hadamard matrices in all regimes, and improves all known transforms in some regimes.

\begin{theorem}\label{thm:uniform-main}
There is an absolute constant $C_{\ref{thm:uniform-main}} > 0$ for which the following holds. Let $d, n, T, \epsilon > 0$ satisfy $n \ge d \ge C_{\ref{thm:uniform-main}}$, $\epsilon \in (0, C_{\ref{thm:uniform-main}}^{-1})$, and  $\epsilon^{-2}\log n \le d$. 
Then, \cref{alg:uniform} runs in time 
\[T \le C_{\ref{thm:uniform-main}}\left(d\log d + \min\left\{d\log n,\epsilon^{-2}(\log n)^2(\log\log n)^2(\log d)^3\right\}\right),\]
and outputs a linear map $\Psi: \mb{R}^d \to \mb{R}^k$, where $k = k(n,\epsilon) \le C_{\ref{thm:uniform-main}}\epsilon^{-2}\log n$, such that for \emph{any} fixed set $X \subseteq \mb{R}^d$ of size $|X|=n$, the inequalities
\[ (1-\epsilon)\|x\|_{2} \le  \|\Psi x\|_{2} \le (1+\epsilon)\|x\|_{2}\]
hold simultaneously for all $x \in X$ with probability at least $2/3$.

Furthermore, for any $x \in \mb{R}^d$, $\Psi x$ can be computed in time $O(T)$ with only $O(1)$ additional memory.
\end{theorem}
\begin{remark}
The probability of failure can be improved to $1-\eta$ with easy and standard modifications of the proof. For the sake of simplicity, we do not keep track of the dependence on $\eta$. Also, the restriction $\epsilon^{-2}\log n\le d$ is of no significance since when $d\le \epsilon^{-2}\log n$, one may simply take $\psi$ to be the identity map.
\end{remark}

As a direct corollary, we obtain a much simpler construction than in \cite{AR14} of a fast RIP-optimal transformation for $s$ matching the restriction in \cref{eqn:ah}. 

\begin{corollary}\label{cor:RIP-optimal}
Fix $\eta > 0$. Given $\delta\in(0,C_{\ref{cor:RIP-optimal}}^{-1})$, $d$, and $s \le \tilde{O}(\delta d^{1/2})$, let $k = C_{\ref{cor:RIP-optimal}}\epsilon^{-2}s\log d$. Define $\Psi$ as in \cref{alg:uniform}, where $n = d^s(1+2/\delta)^s/s!$. Then with probability at least $2/3$, $\Psi$ is RIP of order $s$ and level $C_{\ref{cor:RIP-optimal}}\delta$ and has image dimension at most $k$. Furthermore, application of $\Psi$ on a given point takes $O(d\log d)$ time.
\end{corollary}

\begin{remark}
Unlike approaches based on the Fast Walsh-Hadamard Transform (e.g. \cite{AC09, AR14, bamberger2017optimal}), \cref{alg:uniform} requires neither any preconditioning by random signed diagonal matrices, nor any postconditioning by random permutation matrices. 
\end{remark}

Finally, we obtain similar results, even after replacing the Kac walk by the simpler ORA. 
\begin{theorem}\label{thm:discrete-main}
There is an absolute constant $C_{\ref{thm:discrete-main}} > 0$ for which the following holds. Let $d, n, T, \epsilon > 0$ satisfy $n \ge d \ge C_{\ref{thm:discrete-main}}$, $\epsilon \in (0, C_{\ref{thm:discrete-main}}^{-1})$, and  $\epsilon^{-2}\log n \le d$. 
Then, \cref{alg:averaging} runs in time 
\[T \le C_{\ref{thm:discrete-main}}\left(d\log d \log\log d+ \min\left\{d\log d \log n,\epsilon^{-2}(\log n)^2(\log\log n)^2(\log d)^4\right\}\right),\]
and outputs a linear map $\Psi: \mb{R}^d \to \mb{R}^k$, where $k = k(n,\epsilon) \le C_{\ref{thm:discrete-main}}\epsilon^{-2}\log n$, such that for \emph{any} fixed set $X \subseteq \mb{R}^d$ of size $|X|=n$, the inequalities
\[ (1-\epsilon)\|x\|_{2} \le  \|\Psi x\|_{2} \le (1+\epsilon)\|x\|_{2}\]
hold simultaneously for all $x \in X$ with probability at least $2/3$.

Furthermore, for any $x \in \mb{R}^d$, $\Psi x$ can be computed in time $O(T)$ with only $O(1)$ additional memory.
\end{theorem}

\subsection{Techniques}
Our algorithms are very simple, and are best viewed as running in two phases: in the first phase, we achieve an embedding into a nearly optimal dimension, and in the second phase, we correct this nearly optimal dimension to the optimal dimension. Despite the simplicity of the algorithms, the analysis is involved, and makes use of a multitude of techniques from probability and high-dimensional geometry.

The analysis of the first phase -- which is simply either the Kac walk on $\mb{R}^{d}$ (with no preconditioning), or the ORA on $\mb{R}^{d}$ preceded by a single preconditioning step -- boils down to two things. First, we need the near-optimal JL property of randomly subsampled rows of `bounded' orthogonal matrices, which is proved via chaining methods (see \cref{thm:dirksen} and \cref{thm:krahmer-ward}). Second, we need to show that $O(d\log{d})$ steps of the Kac walk or ORA lead to sufficiently bounded orthogonal systems. For the Kac walk, this follows by making use of a contractive coupling, introduced in work of Pillai and Smith \cite{PS17} studying the total variation mixing time of the Kac walk on the sphere. Unfortunately, for the ORA, this coupling breaks down (as noted earlier, after any finite number of steps, the ORA distribution has total variation distance $1$ from the uniform distribution on the sphere, which explains some of the difficulty in devising coupling-based arguments); we get around this by a completely different argument based on combining the FKG correlation inequality with a delicate recursive computation of various (weighted) moments of the ORA distribution. Furthermore, we demonstrate various symmetry properties of the Kac walk, which enable us to forego the preconditioning by random signed diagonal matrices and/or postconditioning by random permutation matrices present in previous works in this area \cite{AC09, ailon2009fast, bamberger2017optimal, AR14} -- this involves bringing in tools from nonabelian Fourier analysis, in particular adapting work of Diaconis and Shahshahani \cite{DS81} on the transposition walk on the symmetric group $\mf{S}_d$.

The second phase of our algorithm is either a Kac walk on $\mb{R}^{d'}$ or ORA on $\mb{R}^{d'}$, run for $O(d'\log{n})$ steps, where $d'$ is the intermediate nearly-optimal dimension from the first step. For the Kac walk, the contractive coupling suffices for the analysis, whereas for the ORA, we need an analysis based on combining our moment computations with a theorem of Lata{\l}a on the moments of sums of independent random variables.

We note that a two stage algorithm achieving similar objectives also appears in a recent work of Bamberger and Krahmer \cite{bamberger2017optimal}, although the two stages used in their work are very different from each other; they use randomly subsampled Hadamard matrices in the first stage, and random Gaussian matrices in the second stage, which leads to an error term of $\epsilon^{-2}d\log{n}$. Apart from achieving a better error term of $d\log{n}$, and being much more memory efficient, we additionally show how both of these two seemingly disparate stages can be accomplished by the same process (either Kac walk or ORA).   

Finally, we remark that our analysis of ORA may be of independent interest (for instance, in probability and quantum computing, see \cite{chatterjee2019phase}).

\subsection{Organization}
The rest of this paper is organized as follows. In \cref{sec:uniform}, we present and analyse our Kac-walk based algorithm (one of the proofs, present in \cite{PS17}, is included in \cref{app:proof-contraction} for completeness) and in \cref{sec:ora}, we present and analyse our ORA-based algorithm. \cref{sec:preliminaries} contains some auxiliary results related to removing various preconditioning and postconditioning operations; the proof of one of these results is contained in \cref{app:symmetry}. Finally, \cref{sec:open} contains some open problems and directions for future research.

\section{Preliminaries}
\label{sec:preliminaries}
\subsection{Projection and sampling operators}
Let $\on{Proj}_X$ be the projection operator from $\mb{R}^m$ onto the first $X$ basis vectors; sometimes, $X$ will be a random variable, in which case this is to be understood as first generating $X$, and then projecting onto the first $X$ basis vectors. Let $\on{Proj}_{m,q}$ be a (random) projection from $\mb{R}^m$ to a random subset of basis vectors, where each is kept with probability $q$. Finally, let $\on{Sample}_K$ be a (random) projection to a uniformly random subset of size $K$ of the standard basis vectors. Given a random vector $\xi\in\{\pm 1\}^m$, let $D_\xi$ denote the corresponding diagonal matrix. Note that all of these projections can trivially be computed in $O(m)$ time with $O(1)$ additional space.  

\subsection{Symmetric Kac walks}
We now isolate the notion of symmetric random walks, which do not require any preconditioning by random diagonal Rademacher matrices or postconditioning by random permutation matrices. We begin by noting that both the standard Kac walk as well as ORA are instances of the following more general process.
\begin{definition}\label{def:kac-walk}
For a distribution $q$ on angles $[0,2\pi)$, we define the \emph{$q$-Kac walk} on $\on{SO}(d)$ as follows. Let $Q_0 = I_d\in\mb{R}^{d\times d}$. For all integers $t\ge 1$ sample two distinct uniform random coordinates $i_t,j_t\in[d]$ and a random angle $\theta_t$ from $q$, and let $Q_t = R_{i_t,j_t,\theta_t}Q_{t-1}$, where $R_{i,j,\theta}\in\mb{R}^{d\times d}$ is the rotation in the $(i,j)$ plane given by: 
\begin{align*}
R_{i,j,\theta}(e_k) &= e_k \text{ for } k\notin\{i,j\}\\ 
R_{i,j,\theta}(x_ie_i+x_je_j) &= (x_i\cos\theta - x_j\sin\theta)e_i + (x_i\sin\theta + x_j\cos\theta)e_j.
\end{align*}
By the \emph{q-Kac walk of length} $T$ we mean the random variable $Q_T$.
\end{definition}
Note that the standard Kac walk corresponds to $q$ being the uniform distribution on $[0,2\pi)$, and ORA corresponds to $q$ being the delta distribution concentrated at $\pi/4$. We will refer to the standard Kac walk as simply the Kac walk. 
\begin{definition}
A $q$-Kac walk is said to be symmetric if the distribution $q$ is invariant under the maps $\theta \mapsto -\theta$ and $\theta \mapsto \theta + \pi/2$. 
\end{definition}
Clearly, the Kac walk is symmetric. ORA is not symmetric; however, taking $q$ to be the uniform measure on $\{\pi/4, 3\pi/4, 5\pi/4, 7\pi/4\}$ leads to a symmetric walk, which we call symmetric ORA (S-ORA for short). 

The following two lemmas about symmetric $q$-Kac walks enable us to dispense with various preconditioning/postconditioning operations appearing in the literature. (e.g., in \cite{AC09, AR14, bamberger2017optimal, ailon2009fast}).
\begin{lemma}\label{lem:decomposition}
Consider a uniform vector $\xi\in\{\pm 1\}^d$, conditioned on having product $1$. Then, for any  symmetric $q$-Kac walk,
\[\on{TV}(Q_T,Q_TD_\xi)\le\frac{d\exp\big(-\frac{T}{d-1}\big)}{1-d\exp\big(-\frac{T}{d-1}\big)},\]
where $\on{TV}$ denotes the total variation distance.
\end{lemma}
\begin{proof}
For every pair of distinct indices $i,j\in[d]$, let $D_{i,j}$ be the random diagonal matrix with all $1$s, except in the positions $i,j$, where the entries are either both $1$ or both $-1$ with equal probability. For every time $t\in[T]$, let $D_t$ be a random diagonal matrix distributed as $D_{i_t,j_t}$, all sampled independently from everything except $(i_t,j_t)$.

First, note that $R_{i_t,j_t,\theta_t}$ and $R_{i_t,j_t,\theta_t}D_t$ have the same distribution since our distribution $q$ on angles is invariant under $\theta\leftrightarrow\theta+\pi$. Second, note that the distributions of
\[D_{i',j'}R_{i,j,\theta}\text{ and }R_{i,j,\theta}D_{i',j'}\]
are the same (these being independent random matrices), using the symmetry $\theta\leftrightarrow -\theta$ and
\[D_{i',j'}R_{i,j,\theta}D_{i',j'}^{-1} = R_{i,j,-\theta}\]
in the case when $|\{i,j\}\cap\{i',j'\}| = 1$.

By applying the first operation to $R_{i_T,j_T,\theta_T}$, and then applying the second operation repeatedly to switch the diagonal matrix $D_T$ to the end, we see that $Q_T$ has the same distribution as $Q_TD_T$. Then we do the same with $R_{i_{T-1},j_{T-1},\theta_{T-1}}$, and so on, and thus we have the same distribution as $Q_TD_{T-1}D_T$, and so on, until
\[Q_TD_1\cdots D_T.\]

Now, let $\mc{E}$ be the event that the graph on vertex set $[d]$ spanned by the edges $(i_t,j_t)$ for $t\in[T]$ is connected. Condition on any instantiation of all the pairs $(i_t,j_t)$ such that $\mc{E}$ holds. We easily see that $D_1\cdots D_T$ and $D_\xi$ have the same distribution in this case, and furthermore that $Q_TD_1\cdots D_T$ and $Q_TD_\xi$ also have the same distribution in this case (since after conditioning on our instantiation, $Q_T$ is independent from $D_t$ for $t\in[T]$ as well as $D_\xi$).

Therefore, conditional on $\mc{E}$, we have that $Q_TD_1\cdots D_T$ and $Q_TD_\xi$ have the same distribution, so that
\[\on{TV}(Q_T,Q_TD_\xi) = \on{TV}(Q_TD_1\cdots D_T,Q_TD_\xi)\le\mb{P}[\mc{E}^c].\]
Finally, we have
\[\mb{P}[\mc{E}^c]\le\frac{d\exp\big(-\frac{T}{d-1}\big)}{1-d\exp\big(-\frac{T}{d-1}\big)};\]
this follows from well known results about the $O((\log d)/d)$ threshold for random graphs to be connected \cite[Chapter~7]{Bol01}.
\end{proof}
\begin{remark}
The true cutoff for connectedness occurs at $T = d\log d/2$ and not $T = d\log d$. However, deriving an exact expression suitable for non-asymptotic analysis is nontrivial, and is anyway not a crucial point in our final analysis. 
\end{remark}
In fact, as the next lemma shows, symmetric $q$-Kac walks enjoy a more non-trivial invariance property. Namely, after $O(d\log{d})$ steps, the distribution is essentially invariant under left-multiplication by signed permutation matrices in $\on{SO}(d)$.
This allows us to simplify our transforms further by simply projecting onto an initial segment of coordinates, thus enabling a more straightforward memory-optimal, in-place implementation.  
\begin{lemma}\label{lem:perm-invar}
Fix $d\ge 10$. Let $\Sigma$ be a uniformly chosen signed permutation matrix in $\on{SO}(d)$ and $D_\xi$ be as in \cref{lem:decomposition}. Then, for any symmetric $q$-Kac walk,
\[\on{TV}(Q_T,\Sigma Q_TD_\xi)\le\frac{2d\exp\big(-\frac{T}{d-1}\big)}{1-d\exp\big(-\frac{T}{d-1}\big)} + C_{\ref{lem:perm-invar}}\bigg(d^{1/2}e^{-T/(6d)} + (d!)^{1/2}\bigg(\frac{\sqrt{5}-1}{2}\bigg)^{T/2}\bigg),\]
where $C_{\ref{lem:perm-invar}} > 0$ is an absolute constant.
\end{lemma}
The proof of this result is presented in \cref{app:symmetry}, and relies on character estimates of Diaconis and Shahshahani \cite{DS81} used to prove a sharp cutoff for the transposition walk on $\mf{S}_d$.

\section{Fast JL-Optimal and RIP-Optimal Transforms Using the Kac Walk: Proof of \texorpdfstring{\cref{thm:uniform-main}}{Theorem 1.4} and \texorpdfstring{\cref{cor:RIP-optimal}}{Corollary 1.5}}
\label{sec:uniform}

The proof of \cref{thm:uniform-main} and \cref{cor:RIP-optimal} uses \cref{alg:uniform}.
\begin{algorithm}
\begin{algorithmic}
\STATE{\#Run the uniform Kac walk for $O(d \log d)$ steps}
\STATE{Take $T_1 = 12d\log d$ and $K_1 = \min(d,C_{\ref{alg:uniform}}\epsilon^{-2}\log n(\log\log n)^2(\log d)^3)$. Sample $Q_{T_1}$ from the uniform Kac walk and let
\[\Psi_1 := \sqrt{\frac{d}{K_1}}\cdot\on{Proj}_{\on{Binom}(d,K_1/d)}~\circ~Q_{T_1}.\]}
\STATE{Take $T_2\ge 12K_1\log n$ and $K_2\ge C_{\ref{alg:uniform}}\epsilon^{-2}\log n$. Sample $Q_{T_2}'$ from the uniform Kac walk and let
\[\Psi_2 := \sqrt{\frac{K_1}{K_2}}\cdot\on{Proj}_{\on{Binom}(K_1,K_2/K_1)}~\circ~ Q_{T_2}'.\]}
\STATE{Return \[\Psi = \Psi_2\circ \Psi_1.\]}
\end{algorithmic}
\caption{Fast JL via the Uniform Kac walk}
\label{alg:uniform}
\end{algorithm}

\subsection{Coupling and contraction estimates for the Kac walk on \texorpdfstring{$\mb{S}^{d-1}$}{Sd-1}}
\label{sec:contraction}
In this subsection, we describe a coupling of two copies of the Kac walk $X_t, Y_t$  so that the distance between them goes to zero exponentially quickly -- this is one of the two key steps in our analysis of \cref{alg:uniform}. To begin, note that the Kac walk may be viewed as a discrete-time Markov chain $\{X_{t}\}_{t\ge 0}$ on $\mb{S}^{d-1}$ defined as follows: at every step $t$, choose two coordinates $1\le i_t < j_t \le d$ and an angle $\theta_t \in [0,2\pi)$ uniformly at random, and set
\begin{align}
\label{eqn:kac-walk}
X_{t+1}[i_t] &= \cos(\theta_{t})X_t[i_t] - \sin(\theta_t)X_t[j_t] \nonumber \\
X_{t+1}[j_t] &=  \sin(\theta_t)X_t[i_t] + \cos(\theta_t)X_t[j_t] \nonumber \\
X_{t+1}[k] &= X_{t}[k]\quad k\notin \{i_t,j_t\}.
\end{align}
Let $F: [d] \times [d] \times [0,2\pi) \times \mb{S}^{d-1} \to \mb{S}^{d-1}$ be the map associated with this representation, so that $X_{t+1} = F(i_t, j_t, \theta_t, X_t)$.  

\begin{definition}[Proportional coupling, see Definition 3.1 in \cite{PS17}]
\label{def:proportional-coupling}
Define a coupling of two copies $\{X_t\}_{t\ge 0}, \{Y_t\}_{t\ge 0}$ of Kac's walk as follows. Fix $X_0, Y_0 \in \mb{S}^{d-1}$. Let $(i_0, j_0, \theta_0)$ be the update variables used by $X_1$ in \cref{eqn:kac-walk}. Choose $\varphi \in [0,2\pi)$ uniformly at random among all angles that satisfy
\begin{align*}
    X_1[i_0] &= \sqrt{X_0[i_0]^{2} + X_0[j_0]^{2}}\cos\varphi,\\
    X_i[j_0] &= \sqrt{X_0[i_0]^{2} + X_0[j_0]^{2}}\sin\varphi.
\end{align*}
As noted in \cite{PS17}, if $X_0[i_0] = X_0[j_0] = 0$, then all angles $\varphi$ satisfy this equation; otherwise, there is a unique such $\varphi$, and the value of $\varphi - \theta_0 \mod 2\pi$ does not depend on $\theta_0$. 

Then, choose $\theta_0' \in [0,2\pi)$ uniformly among the angles that satisfy
\begin{align*}
    F(i_0,j_0,\theta'_0,Y_0)][i_0] &= \sqrt{Y_0[i_0]^{2} + Y_0[j_0]^{2}}\cos\varphi,\\
    F(i_0,j_0,\theta'_0, Y_0)[j_0] &= \sqrt{Y_0[i_0]^{2} + Y_0[j_0]^{2}}\sin\varphi,
\end{align*}
and set $Y_1 = F(i_0,j_0,\theta'_0, Y_0)$. Note that this coupling forces $Y_1$ to be as close as possible to $X_1$ in the Euclidean distance (for instance, in two dimensions, we always have $X_1 = Y_1$ under this coupling, and in more than two dimensions, it still forces the points $(0,0), (X_1[i_0], X_1[j_0]), (Y_1[i_0], Y_1[j_0])$ to be collinear).  

Now, continue this process starting from $(X_1,Y_1)$ instead of $(X_0,Y_0)$.
\end{definition}

The following key lemma shows that, under the coupling described above, the distance (interpreted suitably) between two copies of Kac's walk decreases exponentially fast.  
\begin{lemma}[See Lemma 3.3 in \cite{PS17}]
\label{lem:contraction-bound}
Fix $X_0, Y_0 \in \mb{S}^{d-1}$. For $t\ge 0$, couple $(X_{t+1}, Y_{t+1})$ conditional on $(X_t, Y_t)$ according to the coupling in \cref{def:proportional-coupling}. Then, for any $t \ge 0$, Kac's walk on $\mb{S}^{d-1}$ satisfies
\[\mb{E} \bigg[\sum_{i=1}^{d}(X_t[i]^{2} - Y_t[i]^{2})^{2}\bigg] \le 2\bigg(1-\frac{1}{2d}\bigg)^{t}\le 2e^{-t/(2d)}.\]
\end{lemma}
For the reader's convenience, we include the complete (short) proof of this lemma in \cref{app:proof-contraction}.
Given this contractive coupling, we now derive estimates regarding the boundedness of the coordinates of the Kac walk.
\begin{lemma}\label{lem:uniform-walk-bounds}
Fix $X_0\in \mb{S}^{d-1}$. Then, for any $\epsilon\in(1/d,1/2)$, any $K\ge 2$, any $k\in[d]$, and any $t\ge 0$, the (uniform) Kac walk satisfies the following, denoting $X_t = Q_tX_0$.
\begin{enumerate}[1.]
    \item $\mb{P}\bigg[\sum_{i=1}^kX_t[i]^2\not\in\frac{k}{d}[1-\epsilon,1+\epsilon]\bigg]\le 8d^4\exp(-t/(2d))+2\exp(-\epsilon^2k/64)$;
    \item $\mb{P}\bigg[\max_{i,j\in[d]}|Q_t[i,j]|\ge K\sqrt{\frac{\log d}{d}}\bigg]\le 2d^3\exp(-t/(2d)) + 2d^{5/2}\exp(-K^2(\log d)/2)$.
\end{enumerate}
\end{lemma}
\begin{proof}
Let $Y_0$ be a uniformly sampled from the sphere $\mb{S}^{d-1}$ and couple our Kac walk $X_t$ (via the proportional coupling \cref{def:proportional-coupling}) to a Kac walk $Y_t$ starting from $Y_0$. Then, we have
\[\mb{E}\bigg[\bigg|\sum_{i=1}^k(X_t[i]^2-Y_t[i]^2)\bigg|^2\bigg]\le k\mb{E} \bigg[\sum_{i=1}^{d}(X_t[i]^{2} - Y_t[i]^{2})^{2}\bigg] \le 2k\bigg(1-\frac{1}{2d}\bigg)^{t}\le 2ke^{-t/(2d)}\]
by \cref{lem:contraction-bound} and Cauchy--Schwarz. Therefore, Markov's inequality implies that
\begin{equation}\label{eq:coupled-distance}
\mb{P}\bigg[\bigg|\sum_{i=1}^k(X_t[i]^2-Y_t[i]^2)\bigg|\ge\epsilon k/(2d)\bigg]\le 8d^2k^{-1}\epsilon^{-2}e^{-t/(2d)}\le 8d^4e^{-t/(2d)}.
\end{equation}
Given this, it suffices to show that $\sum_{i=1}^kY_t[i]^2$ is well-concentrated, which follows since $Y_t$ is uniformly distributed on $\mb{S}^{d-1}$; we include a short computation demonstrating this well-known fact for completeness. 

Let $c$ be a constant to be specified later. Let $Z$ and $Z'$ be independent uniform random vectors on $\mb{S}^{d-1}$. Then,
\begin{align*}
\mb{P}\bigg[\bigg|\sum_{i=1}^{k}Z[i]^2-\frac{k}{d}\bigg|\ge \frac{\epsilon k}{2d}\bigg]&\le {e^{\frac{-c\epsilon k}{2d}}} \mb{E}\Big[e^{c(\sum_{i=1}^{k}Z[i]^2-\frac{k}{d})} + e^{-c(\sum_{i=1}^{k}Z[i]^2-\frac{k}{d})}\Big]\\
&\le {e^{\frac{-c\epsilon k}{2d}}}\mb{E}\Big[e^{c(\sum_{i=1}^{k}Z[i]^2-Z'[i]^2)} + e^{-c(\sum_{i=1}^{k}Z[i]^2-Z'[i]^2)}\Big]\\
& = 2{e^{\frac{-c\epsilon k}{2d}}}\mb{E}\Big[e^{c(\sum_{i=1}^{k}Z[i]^2-Z'[i]^2)}\Big],\\
\end{align*}
where in the second line we have used Jensen's inequality and in the third line we have used symmetry. Let $r = \sum_{i=1}^{d}G[i]^2$ where $G[i]\sim \mc{N}(0,1/d)$ and let $r'$ be an independent copy of $r$. Using orthogonal invariance of the Gaussian we have that $r\cdot Z$ (i.e.\ we pointwise multiply each coordinate of $Z$ by $r$) is distributed as a Gaussian vector $(G[1],\ldots,G[d])$ with each coordinate distributed as $\mc{N}(0,1/d)$. Using these properties along with $\mb{E}[r] = \mb{E}[r'] = 1$, we see that
\begin{align*}
2e^{-\frac{c\epsilon k}{2d}}\mb{E}\big[e^{c(\sum_{i=1}^{k}Z[i]^2-Z'[i]^2)}\big]&\le 2e^{-\frac{c\epsilon k}{2d}}\mb{E}\big[e^{c(\sum_{i=1}^{k}G[i]^2-G'[i]^2)}\big]\\
&= 2e^{-\frac{c\epsilon k}{2d}}\mb{E}\big[e^{cG^2}\big]^{k}\mb{E}\big[e^{-cG^2}\big]^{k}\\
&= 2e^{-\frac{c\epsilon k}{2d}}(1-4c^2/d^2)^{-k/2},
\end{align*}
where we have use Jensen's inequality to replace $Z[i]^2$ by $(r\cdot Z)[i]^2$ and $Z'[i]^2$ by $(r'\cdot Z')[i]^2$, then independence between coordinates, and then explicit computation (assuming $c < d/2)$. Now let $c = A\epsilon d$, so that ultimately
\[\mb{P}\bigg[\bigg|\sum_{i=1}^{k}Z[i]^2-\frac{k}{d}\bigg|\ge \frac{\epsilon k}{2d}\bigg]\le 2e^{-A\epsilon^2 k/2}(1-4A^2\epsilon^2)^{-k/2}\le 2e^{-A\epsilon^2k/2+4A^2\epsilon^2k}.\]
Finally, letting $A = 1/16$ and union-bounding with \cref{eq:coupled-distance} proves conclusion 1.\ of the lemma.

For the second conclusion, it suffices to prove that
\[\mb{P}\bigg[\max_i|X_t[i]|\ge K\sqrt{\frac{\log d}{d}}\bigg]\le 2d^2\exp(-t/(2d)) + 2d^{3/2}\exp(-K^2(\log d)/2),\]
since then, union bounding over $X_0 = e_1,\ldots,e_d$ immediately gives the desired result.

For this, we that Markov's inequality combined with \cref{lem:contraction-bound} gives
\begin{equation}\label{eq:coupled-distance-2}
\mb{P}\bigg[\max_i|X_t[i]^2-Y_t[i]^2|\ge\frac{\log d}{d}\bigg]\le 2d^2e^{-t/(2d)}.
\end{equation}
Since $Y_t$ is uniformly distributed on the sphere, we have good control over $\max_i|Y_t[i]|^2$. In particular, recall a standard bound on the volume of spherical caps (see e.g., \cite[Lemma~2.2]{Bal97}): for a uniformly random unit vector $Y_t \in \mb{R}^d$ and a basis vector $e_i \in \mb{R}^d$, we have
\begin{equation}\label{eq:cap-volume}
\mb{P}[Y_t[i]\ge\epsilon] = \mb{P}[\abs{Y_t - e_i}^2 \le 2-2\epsilon]\le e^{-d \epsilon^2/2}.
\end{equation}
Similarly, one obtains the same bound for $\mb{P}[Y_{t}[i] \leq -\epsilon]$. Using these two bounds with $\epsilon = \sqrt{(K^2-1)(\log d)/d}$ and taking the union bound over $1\le i\le d$, we see that
\[\mb{P}\bigg[\max_iY_t[i]^2\ge(K^2-1)\frac{\log d}{d}\bigg]\le 2de^{-(K^2-1)(\log d)/2} = 2d^{3/2}e^{-K^2(\log d)/2},\]
which combined with  \cref{eq:coupled-distance-2} gives the desired result.
\end{proof}

\subsection{JL-optimality}
The other key tool in proving \cref{thm:uniform-main} is a (by now) classic result \cite{CT06,RM08,Rau10} that demonstrates the restricted isometry property of orthogonal matrices with `bounded' coordinates. We cite the version due to Dirksen \cite{Dir15}, which provides the best known bounds if one requires the dependence on $\delta$ to be optimal i.e. $\delta^{-2}$.
\begin{theorem}[{\cite[Theorem~4.1]{Dir15}}]\label{thm:dirksen}
Let $U$ be an $N\times N$ orthogonal matrix with $\sup_{i,j\in[N]} \sqrt{N}|U_{i,j}| \le K$. Recall $\delta_s$ is defined as
\[\delta_s(A) = \sup_{\substack{\snorm{x}=1\\s\emph{-sparse}}}|\snorm{Ax}_2^2-1|.\]
Then, $\mb{P}[\delta_s(U_I)\ge\delta]\le\eta$, where $U_I = \sqrt{N/m}\cdot\on{Proj}_{N,q}\circ U$ and $q = m/N$, as long as
\[m\ge C_{\ref{thm:dirksen}} sK^2\delta^{-2}\max((\log s)^2(\log m)(\log N), \log(\eta^{-1})).\]
\end{theorem}
Finally, we need the following slight modification of the previously mentioned result of Krahmer and Ward \cite{KW11} which, along with \cref{lem:decomposition}, will allow us to deduce a Johnson-Lindenstrauss property based on the restricted isometry property of the uniform Kac walk.
\begin{theorem}[{Modified \cite[Theorem~3.1]{KW11}}]\label{thm:krahmer-ward}
Fix $\eta > 0$ and $\epsilon\in(0,1)$, and consider a finite set $E\subseteq\mb{R}^d$ of cardinality $|E| = n$. Set $k\ge C_{\ref{thm:krahmer-ward}}\log(4n/\eta)$, and suppose that $\Phi \in \mb{R}^{m\times d}$ satisfies the Restricted Isometry Property of order $k$ and level $\delta\le\epsilon/4$. Let $\xi\in\mb{R}^d$ be a uniform vector in $\{\pm 1\}^d$, conditioned on having product $1$, and let $D_{\xi}$ denote the $d\times d$ diagonal matrix whose diagonal entries are given by $\xi$. Then, with probability at least $1-\eta$,
\[(1-\epsilon)\snorm{x}_2^2\le\snorm{\Phi D_\xi x}_2^2\le(1+\epsilon)\snorm{x}_2^2\]
uniformly for all $x\in E$.
\end{theorem}
\begin{proof}[Proof sketch]
The proof is identical to the one given in \cite{KW11} once we note that the proof in \cite{KW11} only requires that the vector $\xi$ is distributed as an independent Rademacher vector when restricted to certain \emph{proper} subsets of $[d]$, which this altered random variable clearly satisfies.
\end{proof}

We now have all the tools needed to prove \cref{thm:uniform-main}.

\begin{proof}[Proof of \cref{thm:uniform-main}]
Let $q = K_1/d$. Applying \cref{thm:dirksen} and using the second part of \cref{lem:uniform-walk-bounds} at time $t = T_1 = 12d\log d$, we see that
\[\Psi_1' = \frac{1}{\sqrt{q}}\on{Proj}_{d,q}\circ Q_{T_1},\]
with probability $1-O(1/d)$, satisfies $\mb{P}[\delta_s(\Psi_1')\ge\epsilon/4]\le 1/d$ as long as
\[K_1\gtrsim s(\log d)\epsilon^{-2}(\log s)^2(\log K_1)(\log d).\]
Note in the case $K_1 = d$, the operator $\Psi_1'$ is actually orthogonal.

Now by \cref{thm:krahmer-ward}, we have that if $\delta_s(\Psi_1')\le\epsilon/4$ and $s\ge 40\log(4n/\eta)$, then $\Psi'_{1}\circ D_{\xi}$ acts as a $(1\pm\epsilon)$-isometry on our set of points $X$ with probability at least $1-\eta$. Choosing $\eta = 1/4$ and $s = 40\log(16n)$, we see that this property holds with probability at least $3/4-O(1/d)$ as long as
\[K_1\gtrsim\epsilon^{-2}(\log n)(\log d)^2(\log\log n)^2(\log K_1).\]
Since $K_1\le d$, certainly
\[K_1\gtrsim\epsilon^{-2}(\log n)(\log d)^3(\log\log n)^2\]
suffices. It is easily seen that as long as $C_{\ref{alg:uniform}}$ is chosen large enough, this bound indeed holds for the choice of $K_1$ in \cref{alg:uniform}. 

By Markov's inequality, the probability that the actual number of dimensions in the image of $\on{Proj}_{d,q}$ is more than a constant times $K_1$ is sufficiently small, so with probability at least $5/7$, we have that $\Psi'_{1}\circ D_{\xi}$ is a $(1\pm\epsilon)$-isometry on our points and projects down to at most $O(\epsilon^{-2}(\log n)(\log d)^3(\log\log n)^2)$ dimensions.

Next, we show that $\Psi_2$ is a $(1\pm\epsilon)$-isometry on the image of our point set, $(\Psi_1'\circ D_\xi) X$. In particular, applying the union bound using the first part of \cref{lem:uniform-walk-bounds} over all $n$ vectors in the image immediately gives the desired result as long as $C_{\ref{alg:uniform}}$ is large enough. Since $\Psi_1'\circ D_{\xi},\Psi_2$ are both $(1\pm\epsilon)$-isometries on the relevant sets of points, it follows (after rescaling $\epsilon$) that the composition satisfies the desired isometry property with probability at least, say, $7/10$.

Finally, note that in \cref{alg:uniform}, we use $\Psi_{1}$ instead of the more complicated $\Psi_1' \circ D_{\xi}$ -- that this can be done follows easily from \cref{lem:perm-invar} (and after decreasing the probability of success slightly from $7/10$ to say, $2/3$). 

We now quickly compute the runtime and memory of \cref{alg:uniform}. In order to compute $\Psi_{i}x$ (for $i = 1,2$), we apply rotations $R_{i_t,j_t,\theta_t}$ in sequence, and then sparsify. This clearly requires constant memory as computations can be done in place, and since each $R_{i_t,j_t,\theta_t}$ affects at most $2$ coordinates at once, the runtime is $O(d\log d + K_1\log n)$. 
\end{proof}

\subsection{RIP-optimality}
The proof of
\cref{cor:RIP-optimal} follows exactly as in {\cite[Lemma~5.1]{baraniuk2008simple}}.
\begin{proof}[{Proof sketch of \cref{cor:RIP-optimal}}]
This is an application of \cref{thm:uniform-main}, noting that the size of a $\delta$-net of all $s$-sparse unit vectors in $\mb{R}^{d}$ is at most $\binom{d}{s}(1+2/\delta)^s$.
\end{proof}

\section{Fast JL-Optimal and RIP-Optimal Transforms Using ORA: Proof of \texorpdfstring{\cref{thm:discrete-main}}{Theorem 1.6}}
\label{sec:ora}
The proof of \cref{thm:discrete-main} uses \cref{alg:averaging}.
\begin{algorithm}
\begin{algorithmic}
\STATE{\#Run orthogonal repeated averaging  for $O(d \log d)$ steps}
\STATE{Take $T_1 = C_{\ref{prop:max-coord}}d\log d\log\log d$ and $K_1 = \min(d,C_{\ref{alg:averaging}}\epsilon^{-2}\log n(\log\log n)^2(\log d)^3)$. Sample $Q_{T_1}$ from ORA and $D$ an independent diagonal random Rademacher matrix, and let
\[\Psi_1 := \sqrt{\frac{d}{K_1}}\cdot\on{Proj}_{d,K_1/d}~\circ ~Q_{T_1} \circ D.\]}
\STATE{Take $T_2\ge C_{\ref{prop:moment-mixing}}K_1\log n\log d$ and $K_2\ge C_{\ref{lem:big-data}}\epsilon^{-2}\log n$. Sample $Q_{T_2}'$ and let
\[\Psi_2 := \sqrt{\frac{K_1}{K_2}}\cdot\on{Sample}_{K_1,K_2}~\circ~Q_{T_2}'\circ D'.\]}
\STATE{Return \[\Psi = \Psi_2\circ\Psi_1.\]}
\STATE{\#If ORA is replaced by S-ORA, then $D,D'$ may be omitted and the first projection may be replaced with $\on{Proj}_{\on{Binom}(d,K_1/d)}$ and the second projection with $\on{Proj}_{K_2}$.}
\end{algorithmic}
\caption{Fast JL via ORA}
\label{alg:averaging}
\end{algorithm}

The analysis of \cref{alg:averaging} follows the same high level outline as the analysis of \cref{alg:uniform}. However, due to the unavailability of a tractable contractive coupling between the ORA and uniform distribution on the sphere, the proof of the analogues of \cref{lem:uniform-walk-bounds} is more intricate. We now proceed to the details. 

\begin{definition}
For a vector $\mbf{x}\in\mb{R}^d$ and for $k \in \mb{N} \cup \{0\}$, define
\[S_k(\mbf{x}) := \frac{1}{(2k)!}\sum_{i=1}^dx_i^{2k}.\]
In particular, $S_0(\mbf{x}) = d$. 
\end{definition}
The next simple but crucial lemma studies the evolution of $S_{k}(\mbf{x})$ under one step of ORA. 
\begin{lemma}\label{lem:moment-step}
Let $\mbf{x}$ be an $\mb{S}^{n-1}$-valued random vector, and let $R = R_{i,j,\theta}$ be a random rotation corresponding to a single step of ORA. Then,
\[\mb{E}_{R,\mbf{x}}[(S_k(R\mbf{x}))]\le \bigg(1-\frac{2}{d}\bigg)\mb{E}_{\mbf{x}}[S_k(\mbf{x})] + \frac{2^{1-k}}{d(d-1)}\sum_{a=0}^k\mb{E}_{\mbf{x}}[S_a(\mbf{x})]\mb{E}_{\mbf{x}}[S_{k-a}(\mbf{x})].\]
\end{lemma}
\begin{proof}
By direct computation using the definition of $R$, we have
\begin{align*}
\mb{E}_{R,\mbf{x}}[S_k(R\mbf{x})] &= \bigg(1-\frac{2}{d}\bigg)\mb{E}_{\mbf{x}}[S_k(\mbf{x})] \\ & \qquad + \frac{1}{(2k)!d(d-1)}\bigg(\sum_{i\neq j}\mb{E}_{\mbf{x}}\bigg[\bigg(\frac{x_i+x_j}{\sqrt{2}}\bigg)^{2k}\bigg] + \sum_{i\neq j}\mb{E}_{\mbf{x}}\bigg[\bigg(\frac{x_i-x_j}{\sqrt{2}}\bigg)^{2k}\bigg]\bigg)
\\ &= \bigg(1-\frac{2}{d}\bigg)\mb{E}_{\mbf{x}}[S_k(\mbf{x})] + \frac{1}{d(d-1)}\sum_{i\neq j}\sum_{a=0}^k\frac{2^{1-k}}{(2k)!}\binom{2k}{2a}\mb{E}_{\mbf{x}}[x_i^{2a}x_j^{2k-2a}]
\\ &\le\bigg(1-\frac{2}{d}\bigg)\mb{E}_{\mbf{x}}[S_k(\mbf{x})] + \frac{1}{d(d-1)}\sum_{i\neq j}\sum_{a=0}^k\frac{2^{1-k}}{(2k)!}\binom{2k}{2a}\mb{E}_{\mbf{x}}[x_i^{2a}]\mb{E}_{\mbf{x}}[x_j^{2k-2a}]
\\ &\le\bigg(1-\frac{2}{d}\bigg)\mb{E}_{\mbf{x}}[S_k(\mbf{x})] + \frac{2^{1-k}}{d(d-1)}\sum_{a=0}^k\mb{E}_{\mbf{x}}[S_a(\mbf{x})]\mb{E}_{\mbf{x}}[S_{k-a}(\mbf{x})].
\end{align*}
The first inequality follows from the fact that, conditioned on $\mbf{x}_{-i,-j}$, $(x_i^{2a},x_j^{2k-2a})$ is distributed as $(y^{2a}, (\sqrt{r^{2} - y^{2}})^{2k-2a})$, where $r \ge 0$ is determined by $\mbf{x}_{-i,-j}$, and $y^{2}$ is some distribution (determined by the original distribution on $\mb{S}^{n-1}$ and $\mbf{x}_{-i,-j}$) on the interval $[0,r^{2}]$. 
Since the first coordinate is a non-decreasing function of $y^{2}$ and the second coordinate is a non-increasing function of $y^{2}$, it follows from the FKG inequality that
\[\mb{E}[y^{2a}(r^2-y^2)^{k-a}]\le\mb{E}[y^{2a}]\mb{E}[(r^2-y^2)^{k-a}].\qedhere\]
\end{proof}
From this lemma and a careful computation, one can deduce the following upper bound on the $p$-th moments of the coordinates of $\mbf{x}$.
\begin{proposition}\label{prop:moment-mixing}
There exists an absolute constant $C_{\ref{prop:moment-mixing}}$ for which the following holds. 
Let $\mbf{x}$ be an $\mb{S}^{n-1}$ distributed random vector (in particular, $\mbf{x}$ can be deterministic).
Fix a dimension $d\ge 25$, a positive integer $p\le d$, and consider a time $t\ge C_{\ref{prop:moment-mixing}}pd\log d$. Then,
\[\mb{E}_{Q_{t}, \mbf{x}}[S_p(Q_t\mbf{x})]\le\frac{2^{p-2}d^{1-p}}{p!}.\]
\end{proposition}
\begin{remark}
The proof below shows that taking $C_{\ref{prop:moment-mixing}} = 2.25$ is sufficient. 
\end{remark}
\begin{proof}
We will prove this by strong induction on $p\ge 1$. Also, for lightness of notation, we will omit subscripts in the expectation. 

For $p =1$, note that $S_1(Q_{t}\mbf{x}) = 1/2$ deterministically, so that the assertion holds. Hence, let $p\ge 2$, and suppose we know the statement for $1,\ldots,p-1$. Let $e_{q,t} = \mb{E}[S_q(Q_t\mbf{x})]$.

Let $t' = C_{\ref{prop:moment-mixing}}(p-1)d\log d$, and note that $Q_{t}\mbf{x} = Q_{t-t'}(Q_{t'}\mbf{x}) \sim Q_{t-t'}\mbf{y}$, where $\mbf{y}$ is an $\mb{S}^{n-1}$-valued random vector distributed as $Q_{t'}\mbf{x}$. Hence, by the inductive hypothesis, we have that for all $t\ge t'$ and $0\le q\le p-1$,
\[e_{q,t} = \mb{E}[S_q(Q_{t-t'}\mbf{y})]\le\frac{2^{q-2}d^{1-q}}{q!}.\]
Therefore, by \cref{lem:moment-step} and the above, we have for $t\ge t'$ that
\begin{align}
e_{p,t+1} &= \mb{E}[S_p(Q_{t+1-t'}\mbf{y})]\\
&\le\bigg(1-\frac{2}{d}\bigg)\mb{E}[S_p(Q_{t-t'}\mbf{y})] + \frac{2^{1-p}}{d(d-1)}\sum_{a=0}^p\mb{E}[S_a(Q_{t-t'}\mbf{y})]\mb{E}[S_{p-a}(Q_{t-t'}\mbf{y})]\notag
\\ &\le\bigg(1-\frac{2}{d}+\frac{2^{2-p}}{d-1}\bigg)\mb{E}[S_p(Q_{t-t'}\mbf{y})] + \bigg(\frac{2^{1-p}}{d(d-1)}\sum_{a=1}^{p-1}\frac{2^{p-4}d^{2-p}}{a!(p-a)!}\bigg)\notag
\\ &\le\bigg(1-\frac{2}{d}+\frac{2^{2-p}}{d-1}\bigg)\mb{E}[S_p(Q_{t-t'}\mbf{y})] + \bigg(\frac{2^{-3}}{d-1}\frac{d^{1-p}(2^p-2)}{p!}\bigg)\notag
\\&=\bigg(1-\frac{2}{d}+\frac{2^{2-p}}{d-1}\bigg)e_{p,t}+\bigg(\frac{2^{-3}}{d-1}\frac{d^{1-p}(2^p-2)}{p!}\bigg).\label{eq:iteration-1}
\end{align}
To leverage the above relation, we also need to upper bound $e_{p,t'} = \mb{E}[S_p(\mbf{y})]$. Indeed, by the inductive hypothesis, and the fact that each coordinate of $\mbf{y}$ is bounded in absolute value by $1$, it follows that
\begin{equation}\label{eq:initial-1}
e_{p,t'} = \mb{E}[S_p(\mbf{y})]\le\frac{1}{(2p)(2p-1)}\mb{E}[S_{p-1}(\mbf{y})]\le\frac{2^{p}\cdot d^{2-p}}{p!(2p-1)}.
\end{equation}
To summarize, \cref{eq:iteration-1,eq:initial-1} demonstrate that
\begin{align*}
e_{p,t+1}&\le\bigg(1-\frac{2}{d}+\frac{2^{2-p}}{d-1}\bigg)e_{p,t} + \bigg(\frac{2^{-3}}{d-1}\frac{d^{1-p}(2^p-2)}{p!}\bigg).\\
e_{p,t'}&\le\frac{2^{p}\cdot d^{2-p}}{p!(2p-1)}.
\end{align*}
Since $p\ge 2$ and $d\ge 25$, we have that
$$\left(1-\frac{2}{d} + \frac{2^{2-p}}{d-1}\right) \le \left(1-\frac{2}{d}+\frac{1}{d-1}\right) \le \left(1-\frac{23}{24d}\right).$$
Therefore, by iterating the above relations, we have for $t > t'$ that
\begin{align*}
    e_{p,t} 
    &\le \left(1-\frac{23}{24d}\right)^{t-t'}\cdot \frac{2^{p}d^{2-p}}{p!(2p-1)} + \\
    &\quad + \left(\frac{2^{-3}}{d-1}\frac{d^{1-p}(2^p - 2)}{p!}\right)\sum_{j=0}^{t-t'-1}\left(1-\frac{23}{24d}\right)^{j} \\
    &\le e^{-23(t-t')/24d}\cdot\frac{2^{p}d^{2-p}}{p!(2p-1)} + \frac{2^{p-3}d^{1-p}}{p!}\cdot \frac{25}{23}.
\end{align*}
In particular, for $t - t' \ge 48d\log{d}/23,$ we see that
\begin{align*}
e_{p,t} 
&\le \frac{2^{p-2}d^{1-p}}{p!\cdot 18} + \frac{2^{p-2}d^{1-p}}{p!}\cdot \frac{25}{46}\\
&\le \frac{2^{p-2}d^{1-p}}{p!},
\end{align*}
which completes the inductive step. \qedhere 


\end{proof}
We will also need the following estimate regarding the maximum coordinate of $Q_t\mbf{x}$; this estimate is better than simply applying Markov's inequality to  \cref{prop:moment-mixing}.
\begin{proposition}\label{prop:max-coord}
Fix a vector $\mbf{x}\in \mb{S}^{d-1}$. Let $t \ge C_{\ref{prop:max-coord}}d\log d\log \log d$. Then,
\[\mb{P}\bigg[\snorm{Q_t\mbf{x}}_{\infty}\ge 10\sqrt{\frac{\log d}{d}}\bigg]\le d^{-2}.\]
\end{proposition}
\begin{proof}[Proof of \cref{prop:max-coord}]
We may assume that $d\ge 10^2$ as otherwise, the desired conclusion holds trivially. 

We will show the following: for any $p \in [d]$, there exists a collection of events $A_1,\ldots,A_p$ such that the following holds:
\begin{enumerate}
    \item $A_1 \subseteq \ldots \subseteq A_p$;
    \item $A_{p}$ depends only on the randomness used to generate the ORA for the first $C_{\ref{prop:max-coord}}(d\log d\log(p-1) + d(p-1))$ steps;
    \item  $\mb{P}[A_j] \le \frac{j}{d^5}$ for $1\le j\le p$;
    \item For any $t\ge C_{\ref{prop:max-coord}}(d\log d\log p + dp)$,
    \begin{equation}\label{eq:induction-1b}
\mb{E}[S_p(Q_t\mbf{x})|A_p^{c}]\le5^{p-1}\cdot \frac{d^{1-p}}{2^{p}\cdot p!}.
\end{equation}
\end{enumerate}
We prove this by strong induction on $p\ge 1$. For $p = 1$, we simply set $A_1 = \emptyset$ and note that $S_1(Q_{t}\mbf{x}) = 1/2$ deterministically for all times $t$, so that the requirements for $A_1$ are trivially satisfied. Now suppose $p\ge 2$, and we know the statement for $0,\ldots,p-1$. 

Let $B_p$ be the event that at time $t' = C_{\ref{prop:max-coord}}(d\log d\log (p-1)+d(p-1))$, we have
\[S_{p-1}(Q_{t'}\mbf{x})\ge d^5 \cdot5^{p-2}\frac{d^{2-p}}{2^{p-1}(p-1)!}.\]
Clearly, $B_p$ only depends on the randomness used to generate the first $C_{\ref{prop:max-coord}}d\log d\log (p-1)$ steps. Moreover,
by Markov's inequality and the inductive hypothesis, we have that
\[\mb{P}[B_p|A_{p-1}^c]\le \frac{1}{d^5}\]
and therefore, if we set $A_p = B_p\cup A_{p-1}$ then $A_p$ satisfies the first three conclusions of the inductive hypothesis. To complete the inductive step, we only need to verify the last conclusion.

For this, we begin by noting that deterministically under $A_p^c$, 
\[\snorm{Q_{t'}\mbf{x}}_\infty\le L_p := 2\sqrt{2p-2}d^\frac{7-p}{2(p-1)}.\] The key feature of this bound that we need is that $L_p = \sqrt{p}d^{-1/2 + \Theta(1/p)}$. Thus, by the induction hypothesis,
\begin{align}
\mb{E}[S_p(Q_{t'}\mbf{x})|A_p^{c}] &\le \frac{L_{p}^2}{(2p)(2p-1)}\mb{E}[S_{p-1}(Q_{t'}\mbf{x})|A_p^{c}]\notag\\
&\le \frac{L_{p}^2}{(2p)(2p-1)}\mb{E}[S_{p-1}(Q_{t'}\mbf{x})|A_{p-1}^{c}]\cdot \left(1+\frac{1}{d^{4}}\right)\notag\\
&\le 5^{p+2}\frac{d^{1-p+6/(p-1)}}{2^pp!}.\label{eq:initial-1b}
\end{align}
Let $e_{q,t} = \mb{E}[S_q(Q_t\mbf{x})|A_p^c]$.
For $t\ge t'$ the distribution of $Q_t\mbf{x}$ is the same as the distribution of $Q_{t-t'}\mbf{y}$, where $\mbf{y}$ is an $\mb{S}^{n-1}$-valued random vector distributed as $Q_{t'}\mbf{x}$. Also, by the inductive hypothesis, we have that for all $t\ge t'$ and $1\le q\le p-1$,
\[e_{q,t} = \mb{E}[S_q(Q_{t-t'}\mbf{y})|A_p^{c}]\le 5^{q-2}\frac{d^{1-q}}{2^qq!}\bigg(1+\frac{1}{d^4}\bigg),\] where, as before, the final factor comes from conditioning on $A_p^{c}$ and not $A_q^{c}$. Therefore, by a trivial modification of \cref{lem:moment-step}, we have for $t\ge t'$ that
\begin{align}
e_{p,t+1} = \mb{E}[S_p(Q_{t+1-t'}\mbf{y})|A_p^{c}]&\le\bigg(1-\frac{2}{d}\bigg)\mb{E}[S_p(Q_{t-t'}\mbf{y})|A_p^{c}] \notag
\\&\qquad + \frac{2^{1-p}}{d(d-1)}\sum_{a=0}^p\mb{E}[S_a(Q_{t-t'}\mbf{y})|A_p^{c}]\mb{E}[S_{p-a}(Q_{t-t'}\mbf{y})|A_p^{c}]\notag
\\ &\le\bigg(1-\frac{2}{d}+\frac{2^{2-p}}{d-1}\bigg)\mb{E}[S_p(Q_{t-t'}\mbf{y})|A_p^{c}]\notag
\\ &\qquad + 5^{p-2}\bigg(\frac{2^{1-p}}{d(d-1)}\sum_{a=1}^{p-1}\frac{d^{2-p}}{2^pa!(p-a)!}\bigg)\bigg(1+\frac{1}{d^4}\bigg)^2\notag
\\ &=\bigg(1-\frac{2}{d}+\frac{2^{2-p}}{d-1}\bigg)\mb{E}[S_p(Q_{t-t'}\mbf{y})|A_p^{c}]\notag
\\ &\qquad + 5^{p-2}\bigg(\frac{2^{1-p}}{d-1}\frac{d^{1-p}(2^p-2)}{2^pp!}\bigg)\bigg(1+\frac{1}{d^4}\bigg)^2\notag
\\&=\bigg(1-\frac{2}{d}+\frac{2^{2-p}}{d-1}\bigg)e_{p,t}+5^{p-2}\bigg(\frac{2^{1-p}}{d-1}\frac{d^{1-p}(2^p-2)}{2^pp!}\bigg)\bigg(1+\frac{1}{d^4}\bigg)^2.\label{eq:iteration-1b}
\end{align}
To summarize, \cref{eq:iteration-1b,eq:initial-1b} demonstrate that
\begin{align*}
e_{p,t+1}&\le\bigg(1-\frac{2}{d}+\frac{2^{2-p}}{d-1}\bigg)e_{p,t} + 5^{p-2}\bigg(\frac{2^{1-p}}{d-1}\frac{d^{1-p}(2^p-2)}{2^pp!}\bigg)\bigg(1+\frac{1}{d^4}\bigg)^2,\\
e_{p,t'}&\le 5^{p+2}\frac{d^{1-p+6/(p-1)}}{2^pp!}.
\end{align*}
Now, a very similar computation to the one in the proof of \cref{prop:moment-mixing} shows that for 
$$ t-t' \ge O(1)\left(\frac{d\log{d}}{p} + d\right),$$
$e_{p,t} \le 5^{p-1}\cdot\frac{d^{1-p}}{2^{p}\cdot p!},$
which completes the inductive step.\\ 



The proof of the conclusion of \cref{prop:max-coord} now follows easily. Indeed, take $p = 10 \log d$, and note that $\mb{P}[A_p]\le\frac{1}{d^4}$ and that for $t' \ge  C_{\ref{prop:max-coord}}d\log d \log \log d$ Markov's inequality applied to \cref{eq:induction-1b} yields
\[\mb{P}\bigg[S_p(Q_{t'}\mbf{x})\le d^4\cdot 5^{p-1}\frac{d^{1-p}}{2^pp!}\bigg|A_p^{c}\bigg]\le\frac{1}{d^4}.\]
Trivial estimation based on $S_p(\mbf{x})\ge\snorm{\mbf{x}}_\infty^p/(2p)!$ gives the desired result.
\end{proof}

Finally, we prove an estimate which will be required in the second phase of \cref{alg:averaging}. For this, we will make use of the following result of Lata{\l}a \cite{Lat97}.

\begin{lemma}[{\cite[Corollary~2]{Lat97}}]\label{lem:latala}
For a random variable $X$, let $\snorm{X}_s = (\mb{E}|X|^s)^{1/s}$. There exists an absolute constant $C_{\ref{lem:latala}}$ for which the following holds. Let $X_1,\ldots,X_n$ be independent copies of a symmetric random variable $X$. Then,  
\[\snorm{X_1+\ldots+X_n}_p\le C_{\ref{lem:latala}}\sup\bigg\{\frac{p}{s}\bigg(\frac{n}{p}\bigg)^{1/s}\snorm{X}_{s}:\max(2,p/n)\le s\le p\bigg\}.\]
\end{lemma}

\begin{lemma}\label{lem:big-data}
Let $Q_{t}$ denote ORA of length $t$, and let $X_t = Q_tX_0$ with $t\ge C_{\ref{prop:moment-mixing}}d\log d\log n$. Choose a uniformly random set $S$ of indices of size $|S| = k$. If $k = C_{\ref{lem:big-data}}\epsilon^{-2}\log n$, then
\[\mb{P}\bigg[\sum_{i\in S}X_t[i]^2\not\in\frac{k}{d}[1-\epsilon,1+\epsilon]\bigg]\le n^{-3}.\]
\end{lemma}
\begin{proof}
Choose $k$ independent random indices $i_1,\ldots,i_k$, potentially repeated. We first show that for any $p\ge 1$,
\begin{equation}\label{eq:unif-to-indep}
\mb{E}\bigg|\sum_{i\in S}X_t[i]^2-\frac{k}{d}\bigg|^p\le\mb{E}\bigg|\sum_{j=1}^kX_t[i_j]^2-\frac{k}{d}\bigg|^p.
\end{equation}
To see this, consider the joint distribution on $[d]^{k} \times \binom{[d]}{k}$ given by $(i_1,\dots, i_k, T)$, where $i_1,\dots,i_k$ are independent random indices, potentially repeated, and $T$ is a set of size $k$, chosen uniformly at random from among all subsets of $[d]$ of size $k$ containing $\{i_1,\dots, i_k\}$. Note in particular that by symmetry, the marginal distribution of $T$ is  
uniform on $\binom{[d]}{k}$. Therefore, \cref{eq:unif-to-indep} will follow from the law of total probability if we can show that
\[\bigg|\sum_{i\in S}X_t[i]^2-\frac{k}{d}\bigg|^p\le\mb{E}\bigg[\bigg|\sum_{j=1}^kX_t[i_j]^2-\frac{k}{d}\bigg|^p\bigg|T=S\bigg]\]
for all $|S| = k$. But now, notice that the distribution on $(i_1,\ldots,i_k)$ conditioned on $T = S$ is  \emph{some} distribution on $S^k$ which is symmetric under permutations of $S$. Thus, Jensen's inequality immediately implies \cref{eq:unif-to-indep}.

For the remainder of the proof, we will focus on the model with $k$ independent random indices. Let $Y_t = Q'_{t}Y_0$, where $Y_0 = X_0$ and $Q'_{t}$ is an independent copy of $Q_{t}$. We have
\begin{align*}
\bigg(\mb{E}\bigg|\sum_{j=1}^kX_t[i_j]^2-\frac{k}{d}\bigg|^p\bigg)^{1/p}&\le\bigg(\mb{E}\bigg|\sum_{j=1}^kX_t[i_j]^2-Y_t[i_j]^2\bigg|^p\bigg)^{1/p}\\
&\le C_{\ref{lem:latala}}\sup_{2\le s\le p}\frac{p}{s}\bigg(\frac{k}{p}\bigg)^{1/s}(\mb{E}|X_t[i_1]^2-Y_t[i_1]^2|^s)^{1/s}\\
&\le 2C_{\ref{lem:latala}}\sup_{2\le s\le p}\frac{p}{s}\bigg(\frac{k}{p}\bigg)^{1/s}(\mb{E}_{\{Q_{t}, i_1\}}|X_t[i_1]|^{2s})^{1/s}\\
&= 2C_{\ref{lem:latala}}\sup_{2\le s\le p}\frac{p}{s}\bigg(\frac{k}{p}\bigg)^{1/s}\bigg(\mb{E}_{Q_t}\frac{1}{d}\sum_{i=1}^d|X_t[i]|^{2s}\bigg)^{1/s},
\end{align*}
where the first line uses Jensen's inequality, the second line uses \cref{lem:latala}, and the third line uses the triangle inequality.

By \cref{prop:moment-mixing}, if $1\le s\le p$ is an integer, then 
\[\mb{E}_{Q_t}\frac{1}{d}\sum_{i=1}^dX_t[i]^{2s}\le\frac{2^{s-2}(2s)!d^{1-s}}{s!}\]
as long as $t\ge C_{\ref{prop:moment-mixing}}pd\log d$. This (combined with H\"older's inequality to interpolate non-integer moments) shows that
\[\bigg(\mb{E}\bigg|\sum_{j=1}^kX_t[i_j]^2-\frac{k}{d}\bigg|^p\bigg)^{1/p}\le\bigg(\mb{E}\bigg|\sum_{j=1}^kX_t[i_j]^2-Y_t[i_j]^2\bigg|^p\bigg)^{1/p}\le 2C_{\ref{lem:latala}}\sup_{2\le s\le p}\frac{p}{s}\bigg(\frac{k}{p}\bigg)^{1/s}\frac{10s}{d}.\]
Now \cref{eq:unif-to-indep} gives
\[\bigg(\mb{E}\bigg|\sum_{i\in S}X_t[i]^2-\frac{k}{d}\bigg|^p\bigg)^{1/p}\le 2C_{\ref{lem:latala}}\sup_{2\le s\le p}\frac{p}{s}\bigg(\frac{k}{p}\bigg)^{1/s}\frac{10s}{d}.\]
Now, for $k = C\epsilon^{-2}\log n$ and $p = \log n$, we see that the supremum is attained at $s = 2$, so that by Markov's inequality,
\[\mb{P}\bigg[\sum_{i\in S}X_t[i]^2\not\in\frac{k}{d}[1-\epsilon,1+\epsilon]\bigg]\le\bigg(\bigg(\frac{d}{k\epsilon}\bigg)\cdot 20C_{\ref{lem:latala}}p\bigg(\frac{k}{p}\bigg)^{1/2}\frac{1}{d}\bigg)^p.\]
Choosing $C > 10^6C_{\ref{lem:latala}}^2$, we find that this is less than $1/n^3$, as desired.
\end{proof}

We are now ready to prove \cref{thm:discrete-main}.

\begin{proof}[Proof of \cref{thm:discrete-main}]
Let $q = K_1/d$. Applying \cref{thm:dirksen} and using \cref{prop:max-coord} at time $t = T_1 = C_{\ref{prop:max-coord}}d\log d\log\log d$, we see that
\[\widetilde{\Psi_1} = \frac{1}{\sqrt{q}}\on{Proj}_{d,q}\circ Q_{T_1},\]
with probability $1-O(1/d)$, satisfies $\mb{P}[\delta_s(\widetilde{\Psi_1})\ge\epsilon/4]\le 1/d$ as long as
\[K_1\gtrsim s(\log d)\epsilon^{-2}(\log s)^2(\log K_1)(\log d).\]
Note that in the case $K_1 = d$, the operator $\widetilde{\Psi_1}$ is actually orthogonal.

Now by \cref{thm:krahmer-ward}, we have that if 
$s\ge 40\log(4n/\eta)$, then $\Psi_1$ acts as a $(1\pm\epsilon)$-isometry on our set of points $X$ with probability at least $1-\eta$. Choosing $\eta = 1/4$ and $s = 40\log(16n)$, we see that this property holds with probability at least $3/4-O(1/d)$ as long as
\[K_1\gtrsim\epsilon^{-2}(\log n)(\log d)^2(\log\log n)^2(\log K_1).\]
Since $K_1\le d$,
\[K_1\gtrsim\epsilon^{-2}(\log n)(\log d)^3(\log\log n)^2\]
certainly suffices. This indeed holds based on the choice of $K_1$ in \cref{alg:averaging}, as long as $C_{\ref{alg:averaging}}$ is chosen large enough. Note that if we use S-ORA instead of ORA, then by \cref{lem:decomposition}, this holds also for $\widetilde{\Psi_1}$, so that indeed, the random diagonal Rademacher matrix $D$ may be excluded. Furthermore, due to the permutation symmetry in S-ORA established by \cref{lem:perm-invar}, we can replace $\on{Proj}_{d,q}$ in the definition of $\wt{\Psi_1}$ by $\on{Proj}_{\on{Binom}(d,q)}$, similar to the argument in the proof of \cref{thm:uniform-main} (the symmetrization to $\Psi_2$ is similar and we will not further elaborate on this point).

By Markov's inequality, the probability that the actual number of dimensions in the image of $\on{Proj}_{d,q}$ is more than a constant times $K_1$ is sufficiently small, so with probability at least $5/7$ we have that $\Psi_{1}$ is a $(1\pm\epsilon)$-isometry on our points and projects down to at most $O(\epsilon^{-2}(\log n)(\log d)^3(\log\log n)^2)$ dimensions.

To finish, we claim that 
$\Psi_2$ is a $(1\pm\epsilon)$-isometry on the image of our point set, $\Psi_1X$ -- as long as $C_{\ref{alg:averaging}}$ is large enough, this follows immediately by using  \cref{lem:big-data} and taking the union bound over all $n$ vectors in the image. Since $\Psi_1,\Psi_2$ are both $(1\pm\epsilon)$-isometries on the relevant sets of points, we are immediately done (after rescaling $\epsilon$): the desired isometry property holds with probability at least, say, $2/3$.

Finally, the analysis of the running time and memory of \cref{alg:averaging} is essentially identical to that of \cref{alg:uniform}.
\end{proof}
\section{Open Problems}
\label{sec:open}
The most immediate problem left open by our work is to remove the additional $\log\log{d}$ term from \cref{thm:discrete-main}, and bring the ORA-based \cref{alg:averaging} on par with Kac walk and Hadamard matrix based transforms. Another intriguing question is whether algorithms based on the Kac walk can be used to successfully design optimal JL transforms beyond \cref{eqn:bk}/ optimal RIP transforms beyond \cref{eqn:ah}, running in time $O(d\log{d})$; indeed, the appearance of the error term $O(d\log{n})$ in our bounds (as opposed to $O(\epsilon^{-2}d\log{n})$) provides evidence that Kac walk based transforms outperform Hadamard matrix based transforms in large-data/high-accuracy regimes. Finally, it would be very interesting to compare how implementations of Kac walk or ORA-based transforms (optimized for issues/features such as cache locality, parallelization, and memory efficiency) compare to transforms based on Hadamard matrices; see \cite{choromanski2019unifying} for some experimental results in this direction.

\section{Acknowledgements}
We thank Haim Avron and Sourav Chatterjee for helpful comments on an early version of this paper.

\bibliographystyle{amsplain0.bst}
\bibliography{main.bib}
\appendix
\section{Proof of \texorpdfstring{\cref{lem:contraction-bound}}{Lemma 3.2}}
\label{app:proof-contraction}
\begin{proof}
Let $A_{t}[i]=X_{t}[i]^{2}$ and $B_{t}[i]=Y_{t}[i]^{2}$ for all
$t\ge0$ and $i\in[d]$ and recall that $X_t$ and $Y_t$ are coupled as in \cref{def:proportional-coupling}. We calculate,
{\small
\begin{align*}
\mb{E}\left[\sum_{k=1}^{d}(A_{1}[k]-B_{1}[k])^{2}\right] & =\frac{2}{d(d-1)}\sum_{1\le i<j\le d}\mb{E}\left[\sum_{k=1}^{d}(A_{1}[k]-B_{1}[k])^2|(i_{0},j_{0})=(i,j)\right]\\
 & =\frac{2}{d(d-1)}\frac{(d-1)(d-2)}{2}\sum_{k=1}^{d}(A_{0}[k]-B_{0}[k])^{2}\\
 & \quad+\frac{2}{d(d-1)}\sum_{i<j}\mb{E}\left[\left((A_{0}[i]+A_{0}[j])\cos^{2}\varphi-(B_{0}[i]+B_{0}[j])\cos^{2}\varphi\right)^{2}\right]\\
 & \quad+\frac{2}{d(d-1)}\sum_{i<j}\mb{E}\left[\left((A_{0}[i]+A_{0}[j])\sin^{2}\varphi-(B_{0}[i]+B_{0}[j])\sin^{2}\varphi\right)^{2}\right]\\
 & =\frac{d-2}{d}\sum_{k=1}^{d}(A_{0}[k]-B_{0}[k])^{2}\\
 & \quad+\frac{4}{d(d-1)}\mb{E}[\cos^{4}\varphi]\sum_{i<j}\left((A_{0}[i]+A_{0}[j])-(B_{0}[i]+B_{0}[j])\right)^{2}\\
 & =\left(1-\frac{2}{d}\right)\sum_{k=1}^{d}(A_{0}[k]-B_{0}[k])^{2}+\frac{3}{2d(d-1)}\sum_{i<j}\left((A_{0}[i]+A_{0}[j])-(B_{0}[i]+B_{0}[j])\right)^{2}\\
 & =\left(1-\frac{2}{d}\right)\sum_{k=1}^{d}(A_{0}[k]-B_{0}[k])^{2}+\frac{3}{2d(d-1)}\sum_{i<j}\left((A_{0}[i]-B_{0}[i])^{2}+(A_{0}[j]-B_{0}[j])^{2}\right)\\
 & \quad+\frac{3}{d(d-1)}\sum_{i<j}(A_{0}[i]-B_{0}[i])(A_{0}[j]-B_{0}[j])\\
 & =\left(1-\frac{2}{d}\right)\sum_{k=1}^{d}(A_{0}[k]-B_{0}[k])^{2}+\frac{3}{2d}\sum_{k=1}^{d}(A_{0}[k]-B_{0}[k])^{2}\\
 & \quad+\frac{3}{d(d-1)}\sum_{i<j}(A_{0}[i]-B_{0}[i])(A_{0}[j]-B_{0}[j])\\
 & =\left(1-\frac{1}{2d}\right)\sum_{k=1}^{d}(A_{0}[k]-B_{0}[k])^{2}+\frac{3}{d(d-1)}\sum_{i<j}(A_{0}[i]-B_{0}[i])(A_{0}[j]-B_{0}[j])\\
 & =\left(1-\frac{1}{2d}\right)\sum_{k=1}^{d}(A_{0}[k]-B_{0}[k])^{2}\\
 & \quad+\frac{3}{2d(d-1)}\left(\left(\sum_{k=1}^{d}(A_{0}[k]-B_{0}[k])\right)^{2}-\sum_{k=1}^{d}(A_{0}[k]-B_{0}[k])^{2}\right)\\
 & =\left(1-\frac{1}{2d}-\frac{3}{2d(d-1)}\right)\sum_{k=1}^{d}(A_{0}[k]-B_{0}[k])^{2},
\end{align*}
}

where the last equality uses $\sum_{k=1}^{d}A_{0}[k]=1=\sum_{k=1}^{d}B_{0}[k]$.
Thus, we have
\[
\mb{E}\left[\sum_{k=1}^{d}(A_{1}[k]-B_{1}[k])^{2}\right]\le\left(1-\frac{1}{2d}\right).
\]

For $t\ge0$, let $\mathcal{F}_{t}$ denote the $\sigma$-algebra generated
by the random variables $X_{0},\dots,X_{t}$ and $Y_{0},\dots,Y_{t}$.
Repeatedly applying the previous inequality, we have for all $t\ge0$
that 
\begin{align*}
\mb{E}\left[\sum_{k=1}^{d}\left(A_{t}[k]-B_{t}[k]\right)^{2}\right] & =\mb{E}\left[\mb{E}\left[\sum_{k=1}^{d}(A_{t}[k]-B_{t}[k])^{2}\mid\mathcal{F}_{t-1}\right]\right]\\
 & \le\left(1-\frac{1}{2d}\right)\mb{E}\left[\sum_{k=1}^{d}(A_{t-1}[k]-B_{t-1}[k])^{2}\right]\\
 & \le\left(1-\frac{1}{2d}\right)^{t}\sum_{k=1}^{d}\mb{E}\left[(A_{0}[k]-B_{0}[k])^{2}\right]\\
 & \le2\left(1-\frac{1}{2d}\right)^{t}
\end{align*}
as desired.
\end{proof}

\section{Proof of \cref{lem:perm-invar}}\label{app:symmetry}
Throughout this section, we will freely use various tools from nonabelian Fourier analysis; we refer the reader to \cite{DS81} for an introduction to such techniques. We will let $\rho$ denote an irreducible representation of $\mf{S}_d$, i.e., $\rho\in\wh{\mf{S}_d}$ and $d_{\rho}$ denote its dimension. Since $\mf{S}_{d}$ is finite, all its finite dimensional representations are unitarizable, and we will work with a choice of inner product such that irreducible representations are also unitary. In particular, various appearances of $\dagger$ should be understood as the operator-theoretic adjoint with respect to the appropriate inner product. A sum over nontrivial irreducible representations will be denoted $\sum_\rho'$. 
The key estimate we need is the following purely probabilistic claim regarding permutations.
\begin{lemma}\label{lem:perm-mixing}
Let $\sigma=(\sigma_1,\ldots,\sigma_T)$ be uniformly randomly chosen transpositions in $\mf{S}_d$ and suppose $d\ge 10$. Let $\xi_i\sim\on{Ber}(1/2)$ for $1\le i\le T$. Define
\[P_\sigma = \sigma_1^{\xi_1}\cdots\sigma_T^{\xi_T}.\]
Then
\[\mb{E}_{\sigma, \xi}[\on{TV}(P_\sigma,\on{Unif}_{\mf{S}_d})]\le C_{\ref{lem:perm-mixing}}\bigg(d^{1/2}e^{-T/(6d)} + (d!)^{1/2}\bigg(\frac{\sqrt{5}-1}{2}\bigg)^{T/2}\bigg)\]
for an absolute constant $C_{\ref{lem:perm-mixing}} > 0$.
\end{lemma}
\begin{remark}
The given proof can be modified (with more careful character estimates similar to \cite{DS81}) to show the quantity studied tends to $0$ once $T$ passes $2d\log d$ (with a $\Theta(d)$ rate). It is an interesting question as to whether this is the sharp cutoff.
\end{remark}
\begin{proof}
Let $U: \mf{S}_d\to\mb{C}$ be $1/d!$ everywhere. For any permutation $\tau$, let $f_\tau: \mf{S}_d\to\mb{C}$ be $1/2$ at the identity and $\tau$, and $0$ elsewhere. We note that $\wh{U}(\rho) = 0$ for nontrivial representations $\rho$. We also note that
\[\mb{P}[P_\sigma=\tau] = f_{\sigma_1}\ast\cdots\ast f_{\sigma_T}(\tau)\]
by the definition of convolution. The Fourier coefficient of this function at $\rho$ is
\[A_\sigma(\rho):= \wh{f_{\sigma_1}}(\rho)\cdots\wh{f_{\sigma_T}}(\rho).\]

By the proof of the upper bound lemma of Diaconis and Shahshahani \cite{DS81}, we have
\begin{align*}
\mb{E}_\sigma[\on{TV}(P_\sigma,\on{Unif}_{\mf{S}_d})] &= \mb{E}_\sigma\bigg[\sum_{\tau\in\mf{S}_d}\bigg|\mb{P}[P_\sigma=\tau]-\frac{1}{d!}\bigg|\bigg]\\
&\le\bigg[\mb{E}_\sigma\bigg[ d!\sum_{\tau\in\mf{S}_d}\bigg(\mb{P}[P_\sigma=\tau]-\frac{1}{d!}\bigg)^2\bigg]\bigg]^{1/2}\\
&=\bigg[\mb{E}_\sigma\bigg[{\sideset{}{'}\sum_\rho} d_\rho\on{Tr}(A_\sigma(\rho)A_\sigma(\rho)^\dagger)\bigg]\bigg]^{1/2},
\end{align*}
where the first line is by definition, the second line uses Cauchy--Schwarz, and the third line uses Plancherel's formula. In the third line, we also used that $U$ has zero Fourier coefficient at nontrivial representations, and that the term at the trivial representation cancels out.

Next, we claim that
\[\mb{E}_{\sigma_i}\bigg[\wh{f_{\sigma_i}}(\rho)\wh{f_{\sigma_i}}(\rho)^\dagger\bigg] = c_\rho I_{d_\rho}\]
for a constant $c_\rho\in\mb{R}$. In fact, we can compute this constant explicitly. Let $\chi_\rho$ be the trace of $\rho$ evaluated at any transposition, and let $r_\rho = \chi_\rho/d_\rho$. It is worth noting that $|r_\rho|\le 1$ since unitary matrices have trace at most $d_\rho$. We find, since $\rho(\sigma_i)$ is unitary, that
\[\mb{E}_{\sigma_i}\bigg[\wh{f_{\sigma_i}}(\rho)\wh{f_{\sigma_i}}(\rho)^\dagger\bigg] = \mb{E}_{\sigma_i}\bigg[\bigg(\frac{I_{d_\rho}+\rho(\sigma_i)}{2}\bigg)\bigg(\frac{I_{d_\rho}+\rho(\sigma_i)}{2}\bigg)^\dagger\bigg] = \frac{1}{2}I_{d_\rho} + \frac{1}{2}\mb{E}_{\sigma_i}[\rho(\sigma_i)] = \bigg(\frac{1+r_\rho}{2}\bigg)I_{d_\rho}.\]
In the last step we noted that $\mb{E}_{\sigma_i}[\rho(\sigma_i)]$ is a multiple of the identity by Schur's lemma (or, it is the Fourier transform of a function constant on conjugacy classes) and has trace $\chi_\rho$ by definition (note that $\chi_{\rho}$ is real since $\rho(\sigma_i)$ is an involution). Thus
\[c_\rho = \frac{1+r_\rho}{2}.\]

Now, note that
\begin{align*}\mb{E}_\sigma\bigg[\on{Tr}(A_\sigma(\rho)A_\sigma(\rho)^\dagger)\bigg] &= \mb{E}_\sigma\bigg[\on{Tr}\bigg(\wh{f_{\sigma_1}}(\rho)\cdots\wh{f_{\sigma_T}}(\rho)\wh{f_{\sigma_T}}(\rho)^\dagger\cdots\wh{f_{\sigma_1}}(\rho)^\dagger\bigg)\bigg]\\
&= c_\rho\mb{E}_{\sigma_1,\ldots,\sigma_{T-1}}\bigg[\on{Tr}\bigg(\wh{f_{\sigma_1}}(\rho)\cdots\wh{f_{\sigma_{T-1}}}(\rho)\wh{f_{\sigma_{T-1}}}(\rho)^\dagger\cdots\wh{f_{\sigma_1}}(\rho)^\dagger\bigg)\bigg]\\
&= \cdots\\
&= d_\rho c_\rho^T = d_\rho\bigg(\frac{1+r_\rho}{2}\bigg)^T.
\end{align*}
Therefore
\[\mb{E}_\sigma[\on{TV}(P_\sigma,\on{Unif}_{\mf{S}_d}]\le\bigg[{\sideset{}{'}\sum_\rho}d_\rho^2\bigg(\frac{1+r_\rho}{2}\bigg)^T\bigg]^{1/2},\]
and it remains to bound the right side.

The key technical result in \cite[p.~27]{DS81} is that
\[\bigg[{\sideset{}{'}\sum_\rho}d_\rho^2\bigg(\frac{1}{d}+\frac{d-1}{d}r_\rho\bigg)^{2k}\bigg]^{1/2}\le Cde^{-2k/d}\]
where $C$ is an absolute constant independent of $d$.
Now if $r_\rho\in[\sqrt{5}-2,1]$ we have
\[0 < \bigg(\frac{1+r_\rho}{2}\bigg)^3\le r_\rho\le\frac{1}{d}+\frac{d-1}{d}r_\rho,\]
while if $r_\rho\in[-1,\sqrt{5}-2]$ we have
\[\bigg|\frac{1+r_\rho}{2}\bigg|\le\frac{\sqrt{5}-1}{2}.\]
Thus, using this, we see that if $6|T$ we have
\begin{align*}
{\sideset{}{'}\sum_\rho}d_\rho^2\bigg(\frac{1+r_\rho}{2}\bigg)^T&\le{\sideset{}{'}\sum_\rho}d_\rho^2\bigg(\frac{1}{d}+\frac{d-1}{d}r_\rho\bigg)^{T/3} + {\sideset{}{'}\sum_\rho}d_\rho^2\bigg(\frac{\sqrt{5}-1}{2}\bigg)^T\\
&\le Cde^{-T/(3d)} + d!\bigg(\frac{\sqrt{5}-1}{2}\bigg)^T.
\end{align*}
Since the $\on{TV}$ is decreasing as $T$ increases, the result follows immediately by rounding $T$ to the nearest multiple of $6$.
\end{proof}
Now we are ready to prove \cref{lem:perm-invar}.
\begin{proof}[Proof of \cref{lem:perm-invar}]
By two applications of \cref{lem:decomposition}, we see that
\begin{align}
\label{eq:sign-removal}
\on{TV}(Q_T,D_\xi Q_TD_{\xi'})\le \frac{2d\exp\big(-\frac{T}{d-1}\big)}{1-d\exp\big(-\frac{T}{d-1}\big)}
\end{align}
if $\xi,\xi'$ are independent random vectors which are uniform over $\{\pm 1\}^d$, conditioned on having product $1$.

Now, let $Q_T = R_{i_T,j_T,\theta_T}\cdots R_{i_1,j_1,\theta_1}$ as usual. For every pair of distinct indices $i,j\in[d]$, let $D_{i,j}$ be the random rotation in the $(i,j)$ plane by a uniform multiple of $\pi/2$. For every time $t\in[T]$, let $D_t$ be a random matrix distributed as $D_{i_t,j_t}$, sampled independently from everything except $(i_t,j_t)$. First, note that $R_{i_t,j_t,\theta_t}$ and $R_{i_t,j_t,\theta_t}D_t$ have the same distribution since our distribution $q$ on angles is invariant under $\theta\leftrightarrow\theta+k\pi/2$ for all $k\in\mb{Z}$. Second, note that the distributions
\[D_{i',j'}R_{i,j,\theta}\text{ and }R_{i,j,\theta}D_{i',j'}\]
are the same. The reason is more subtle than in the proof of \cref{lem:decomposition}. The point is that $D_{i',j'}$ merely permutes and signs the basis vectors $e_1,\ldots,e_d$ (via at worst a transposition). Thus conjugation of $R_{i,j,\theta}$ by $D_{i',j'}$ gives another rotation in a coordinate plane $(\sigma(i),\sigma(j))$ (where $\sigma$ is either the identity or the swap $(i'j')$), with its angle potentially changed via negation, addition by $\pi$, or both. Either way, we see $D_{i',j'}R_{i,j,\theta}D_{i',j'}^{-1}$ (conditional on the value $D_{i',j'}$) has the same distribution as $R_{i,j,\theta}$, hence the claim.

Now we extract the matrices $D_t$ similar to in the proof of \cref{lem:decomposition}. However, we must be slightly careful: note that $D_t$ is dependent on $(i_t,j_t)$, and the swapping operation above can potentially change a pair $(i_t,j_t)$ as we move past (which was not true before). Therefore, we will perform swaps in a way such that once $D_t$ has been extracted to the end, the rotation $R_{i_t,j_t,\theta_t}$ is not touched again. In fact, we were careful to do this already in the proof of \cref{lem:decomposition}, although this care was not needed there.

Specifically, we apply the first operation to $R_{i_T,j_T,\theta_T}$, and then apply the second operation repeatedly to switch the diagonal matrix $D_T$ to the end. Then we do the same for $R_{i_{T-1},j_{T-1},\theta_{T-1}}$, and so on. We thus see that $Q_T$ has the same distribution as
\[Q_TD_1\cdots D_T.\]
Let $\sigma_t = (i_tj_t)$ for $1\le t\le T$. Note that $D_1\cdots D_T$ is independent of $Q_T$ conditional on $\sigma=(\sigma_1,\ldots,\sigma_T)$, and is a signed permutation matrix with determinant $1$. Therefore it can be written uniquely as $D_1\cdots D_T=PD$, where $P$ is an unsigned permutation matrix and $D$ is a diagonal sign matrix, with $\det(PD) = 1$. Note that $(P,D)$ is independent of $Q_T$ conditional on $\sigma$.

Furthermore, we see that we can change $(P,D)$ into a joint distribution on signed permutation matrices (with determinant $1$) and diagonal matrices (with determinant $1$) which has a uniform marginal on $P$ while sacrificing at most a $\on{TV}$ of
\[\on{TV}(P_\sigma,\on{Unif}_{\mf{S}_d}),\]
where $P_\sigma$ is defined as in \cref{lem:perm-mixing}. This is since (conditional on $\sigma$) $D_1\cdots D_T$ induces a permutation on the coordinates $e_1,\ldots,e_d$ with the same distribution as $P_\sigma$.

Let $\Sigma$ be a uniform signed permutation matrix with determinant $1$. We deduce that there is a distribution of diagonal matrices $D$ (with determinant $1$), potentially dependent on $(i_1,j_1),\ldots,(i_T,j_T)$ and $\Sigma$, such that
\[\on{TV}(Q_TD_1\cdots D_T,Q_T\Sigma D)\le\mb{E}_\sigma\on{TV}(P_\sigma,\on{Unif}_{\mf{S}_d}).\]
Therefore, for $D_\xi,D_{\xi'}$ independent from everything as defined at the beginning, we have
\[\on{TV}(D_\xi Q_TD_{\xi'}, D_\xi Q_T\Sigma DD_{\xi'})\le\mb{E}_\sigma\on{TV}(P_\sigma,\on{Unif}_{\mf{S}_d}).\]
Regardless of the value of $D$, we see that the independent sign matrix $D_{\xi'}$ rerandomizes it so that $D_\xi Q_T\Sigma DD_{\xi'}$ and $D_\xi Q_T\Sigma$ have the same distribution. Using this, along with \cref{eq:sign-removal}, we deduce that
\[\on{TV}(Q_T, D_\xi Q_T\Sigma)\le\frac{2d\exp\big(-\frac{T}{d-1}\big)}{1-d\exp\big(-\frac{T}{d-1}\big)} + \mb{E}_\sigma\on{TV}(P_\sigma,\on{Unif}_{\mf{S}_d}),\]
and now \cref{lem:perm-mixing} finishes. Technically, we also note that $Q_T$ is invariant under taking transposes, so that we can also deduce a bound on $\on{TV}(Q_T,\Sigma Q_TD_\xi)$.
\end{proof}

\end{document}